\documentclass[preprint,12pt]{elsarticle}
\nopreprintlinetrue

\usepackage{graphicx}
\usepackage{natbib}
\usepackage{epstopdf}
\usepackage{floatrow}
\usepackage{multirow}
\usepackage{caption}
\usepackage{subcaption}

\usepackage{comment}
\biboptions{sort&compress}

\usepackage{lineno}

\usepackage[applemac]{inputenc}
\usepackage[T1]{fontenc}
\usepackage[ngerman,english]{babel}

\usepackage{xcolor}
\usepackage{enumitem}

\usepackage{amsmath}
\usepackage{amsthm}
\usepackage{amssymb}
\usepackage{mathtools}
\usepackage[colorlinks,allcolors=blue]{hyperref}

\usepackage{hyperref}
\newtheorem{theorem}{Theorem}[section]

\newtheorem{lemma}[theorem]{Lemma}
\newtheorem{assumption}{Assumption}[section]
\newtheorem{remark}{Remark}[section]
\newtheorem{fact}[theorem]{Fact}

\newcommand{\Z}{\mathbb{Z}}
\renewcommand{\u}{\mathbf{u}}

\journal{Physica D: Nonlinear Phenomena}

\begin{document}

\begin{frontmatter}


\title{Long-time behavior of optimal mixing in an advection-diffusion shell model}

\author[label0]{Jiajia Guo\corref{cor1}} 
\ead{jiajiag@andrew.cmu.edu}
\author[label1]{Baole Wen\corref{cor1}}
\ead{bwen@nyit.edu}
\author[label2]{Christian Seis} 
\ead{seis@uni-muenster.de}
\author[label3,label4,label5]{Charles R. Doering}
\affiliation[label0]{organization={Department of Mathematical Sciences, Carnegie Mellon University},
            city={Pittsburgh},
            state={PA},
            postcode={15213},
            country={USA}}      
\affiliation[label1]{organization={Department of Mathematics, New York Institute of Technology}, 
            city={Old Westbury},
            state={NY},
            postcode={11568},
            country={USA}}
\affiliation[label2]{organization={Institut f\"ur Analysis und Numerik, Universit\"at M\"unster},
            postcode={48149 M{\"u}nster},
            country={Germany}}    
\affiliation[label3]{organization={Department of Mathematics, University of Michigan},
            city={Ann Arbor},
            state={MI},
            postcode={48109-1043}, 
            country={USA}}   
\affiliation[label4]{organization={Department of Physics, University of Michigan},
            city={Ann Arbor},
            state={MI},
             postcode={48109-1040}, 
            country={USA}}   
\affiliation[label5]{organization={Center for the Study of Complex Systems, University of Michigan},
            city={Ann Arbor},
            state={MI},
            postcode={48109-1107}, 
            country={USA}}

\cortext[cor1]{Corresponding authors.}

\begin{abstract}
We investigate the long-time behavior of optimal mixing in an advection-diffusion equation using a shell model framework.  Our focus is on quantifying the decay of the scalar variance, measured by the
negative Sobolev norm $H^{-1}$, under enstrophy-constrained stirring.  We perform long-time computations using both local-in-time (maximizing the instantaneous mixing rate) and global-in-time (maximizing mixedness at a prescribed final time) optimization strategies.  For mixing with diffusion ($\kappa>0$), the numerical results show that the scalar length scale eventually becomes limited by a generalized Batchelor scale, in close agreement with theoretical predictions.  In this regime, the $H^{-1}$ mix-norm decays exponentially in time with a decay rate that is independent of the diffusivity $\kappa$.  Compared with the purely advective case ($\kappa = 0$), diffusion significantly enhances the long-time mixing rate; moreover, increasing diffusivity further improves mixing efficiency by reducing the prefactor of the exponential decay.  Guided by these numerical observations, we derive new conditional lower bounds on the $H^{-1}$ norm whose exponential decay rates are strictly independent of the diffusivity parameter $\kappa$, for all $\kappa > 0$.  We further establish conditional upper bounds on the maximal rate of enhanced dissipation of the scalar variance, showing that the effective diffusion time scale is at least of the order $|\log\kappa|$.
\end{abstract}

\end{frontmatter}


\section{Introduction}\label{sec:intro}

Fluid mixing is a ubiquitous phenomenon in nature as well as in industrial and engineering processes, occurring across a wide range of scales, from large-scale atmospheric and oceanic circulation, to medium-scale combustion, to microscale mixing in food and pharmaceutical processing or microfluidic production.  For representative studies, see~Refs.~\citep{Wunsch2004, Annaswamy1995, Hessel2005, Cullen2009}.  A central question in both fundamental research and practical applications concerns the efficiency of mixing, namely, how effectively a flow mixes an observable quantity within a given time or under a fixed resource budget.

Despite the complexity of many systems under consideration, the mathematical model describing fluid mixing is surprisingly simple: it is governed by a linear partial differential equation (PDE), known as the advection-diffusion equation:
\begin{eqnarray}\label{eq:AD}
&    \dfrac{\partial \theta}{\partial t} + \mathbf{u}\cdot \nabla\theta
     = \kappa\Delta\theta.
\end{eqnarray}
Here, $\mathbf{u}(\mathbf{x},t)$ represents a given divergence-free velocity field, $\theta(\mathbf{x},t)$ is the scalar observable, e.g., concentration, mass density, order parameter, or temperature, and $\kappa$ denotes the diffusivity (or thermal diffusivity) of the substance.  Depending on the specific problem, the velocity field $\mathbf{u}$ is often determined by solving a nonlinear PDE, such as the Euler or Navier--Stokes equations.  In some cases, $\mathbf{u}$ is even coupled to the observed quantity $\theta$, for instance, through buoyancy forces as seen in Rayleigh--B\'enard
convection~\citep{Wen2020JFM, Wen2022JFM}.

In this paper, we decouple $\mathbf{u}$ from any such constraints, allowing us to investigate \emph{maximal mixing efficiencies} for the model equation~\eqref{eq:AD} under appropriate integral constraints on the fluid velocity $\mathbf{u}$.  Common choices for these constraints include limiting the kinetic energy via $\|\mathbf{u}\|_{L^2}$ or the enstrophy (viscous dissipation) via $\|\nabla \textbf{u}\|_{L^2}$, see Ref.~\citep{Lin2011}. In the present work, we will solely focus on the finite enstrophy setting.

The fundamental questions that arise in the investigation of maximal mixing efficiencies are twofold: 
\begin{enumerate}[label=(\roman*)]
\item \emph{What is the maximal decay rate of mix-norms under the evolution \eqref{eq:AD} that can be achieved by the constrained velocity fields?}
\item \emph{How can we design an optimally mixing velocity field that saturates the maximal analytical bound?}
\end{enumerate}

In the non-diffusive setting, where $\kappa = 0$, these questions are now quite well understood---at least from an academic standpoint.  In this context, lower bounds, featuring at best exponential mixing, have been established for mix-norms such as the $H^{-1}$ norm \citep{Lunasin2012, Seis2013, Iyer2014}, optimal transportation distances \citep{Seis2013}, and geometric mixing measures \citep{Crippa2008}, {establishing exponential decay as a fundamental barrier that no incompressible flow can exceed.}

In the diffusive setting, where $\kappa > 0$, our rigorous understanding of maximal mixing remains fairly limited.  Since the work of Batchelor \citep{Batchelor1959}, it has been recognized that diffusion limits the effect of a stirring flow beyond a certain length scale, at which diffusive smoothing and filamentation balance each other.  Scales finer than this \emph{Batchelor scale} cannot persist, and at that
point, mixing becomes a matter of reducing intensity differences rather than further reducing scale.  Consequently, it is customary to analyze mixing in the diffusive regime by measuring both the variance $\|\theta\|_{L^2}$ (for mean-zero observables) and a mix-norm like $\|\theta\|_{H^{-1}}$~\citep{Thiffeault2012}.  Both quantities can be estimated to decay at most double exponentially for Lipschitz velocity fields, though exponential bounds are generically expected and supported by numerical evidence \citep{Miles2018}.  Recent work has established conditional bounds on the decay of variance \citep{Seis23, Brue2021, MeyerSeis24}, which are similar in spirit to our new Theorems~\ref{T1} and \ref{T2} below.  These results suggest the validity of exponential mixing rates, but do not rule out the possibility of faster mixing scenarios.  Investigating the decay of the $L^2$ norm is also of interest when comparing it with the classical decay for the heat equation.  In the absence of stirring, i.e., for $\mathbf{u} \equiv 0$, the diffusion time, defined as the time required to reduce the initial variance by a factor of $2$, is of order $1/\kappa$.  In contrast, stirring generates sharp concentration gradients through filamentation, thereby accelerating the dissipation of variance.  Our conditional estimates in section~\ref{sec:UpperBound} show that this enhanced mixing reduces the diffusion time to order $|\log \kappa|$.

In the absence of absolute lower bounds, identifying optimal mixing strategies and thus addressing the second question remains challenging.  On the positive side, exponential mixing rates have been achieved in the diffusive setting for stochastic fluid systems; see, for example, Refs.~\citep{Bedrossian2021, CoopermanIyerSon25, NavarroSeis26}.  These rates are in excellent agreement with the observations in Ref.~\cite{Seis23}, but the question of whether they represent optimal mixing remains unresolved.

One approach to addressing the second question independently of the first is to study the associated optimization problems \citep{Mathew2005, Mathew2007, Lin2011, Miles2018, Miles2018diffusion, Zhu2024}.  These problems can be considered either instantaneously (i.e., locally or pointwise in time) or globally over a given time interval.  While the local-in-time mixing strategy is relatively straightforward, it often leads to sub-optimal bounds.  In contrast, the global-in-time strategy, which aims to maximize mixing at a prescribed final time, presents significant theoretical and numerical
challenges---particularly when the final time is large.

To circumvent the mathematical and numerical difficulties inherent in the full PDE, it is customary to reduce Eq.~\eqref{eq:AD} to a shell model.  Shell models have been extensively used to study energy-cascade mechanisms, energy dissipation, anomalous scaling, and other turbulence phenomena in the context of the Navier--Stokes equations (NSE) \citep{Lorenz1963, Gledzer1973, Yamada1988, Jensen1992, Biferale2003, Constantin2006, Ditlevsen2010}.  While the Fourier transform of the NSE produces an infinite set of coupled ordinary differential equations (ODEs), shell models describe the statistics of turbulence in Fourier space using a simple set of ODEs by ``binning'' the variables with wavenumbers $\sigma^{n-1}k_0 < |\mathbf{k}| < \sigma^n k_0$, where $\sigma>1$ is the intershell ratio and $k_0>0$, into a single variable for each index $n = 1, 2$, $\dots$.  As phenomenological models of turbulence, shell models retain certain features of the original NSE \citep{Constantin2006, Mohamed2018}.

In this investigation, we use a shell model to explore the long-term behavior of optimal mixing under the enstrophy constraint, both with and without diffusion, for local-in-time and global-in-time optimization scenarios.  This extends the earlier work by Miles and Doering~\cite{Miles2018}, which focused on short-term behavior using a shell model.  While studies like those by Lin \emph{et al.}~\citep{Lin2011} and Zhu \emph{et al.}~\citep{Zhu2024} have examined long-term optimal mixing based on the full PDE, their focus has been on the non-diffusive setting under a local-in-time stirring strategy.  

Our analysis and numerical results under the enstrophy-constrained local-in-time strategy confirm that, with diffusion (i.e., $\kappa > 0$), the scalar length scale eventually becomes limited by a generalized Batchelor scale, proportional to $\sqrt{\kappa/\Gamma}$, where $\Gamma$ is the rate-of-strain.  Once this Batchelor regime is reached, the $H^{-1}$ mix-norm $\|\theta(t)\|_{H^{-1}}$ decays exponentially as $t \to \infty$ at a rate that is independent of $\kappa$ and significantly larger than in the non-diffusive setting, demonstrating that diffusion substantially enhances the long-term mixing rate.  These findings are further supported by the numerical results obtained under the enstrophy-constrained global-in-time strategy at large final times; moreover, for both $\kappa = 0$ and $\kappa > 0$, the global-in-time strategy achieves a strictly higher exponential decay rate than the local-in-time strategy.  Guided by these numerical observations, we derive new conditional lower bounds on the mix-norm $\|\theta(t)\|_{H^{-1}}$ for the shell model that are strictly independent of $\kappa$ for all $\kappa > 0$, and we establish conditional upper bounds on the maximal rate of enhanced dissipation that quantify the fundamental role of diffusivity in the long-time mixing dynamics.

The rest of this paper is organized as follows.  Section~\ref{sec:Eqns} outlines the problem formulation for the enstrophy-constrained local-in-time and global-in-time optimizations aimed at maximizing mixing based on the $H^{-1}$ norm for the shell model, and introduces the numerical methods used to solve these optimization problems.  The numerical results are discussed and analyzed in section~\ref{sec:Results}.  Drawing on the numerical evidence of the existence and length of the Batchelor scale in both strategies, we derive conditional lower bounds on the $H^{-1}$ mix-norm $\|\theta(t)\|_{H^{-1}}$ in section~\ref{sec:LowerBound}. 
Subsequently, section~\ref{sec:UpperBound} presents conditional upper bounds on the maximal rate of enhanced dissipation, thus analyzing the decay rates of the $L^2$ norm $\|\theta(t)\|_{L^{2}}$. Finally, our conclusions are summarized in section~\ref{sec:Conclusion}.

\section{Governing equations and numerical methods} \label{sec:Eqns}

Our starting point is the advection-diffusion equation \eqref{eq:AD} in a periodic box $\Omega = [0, W]^d$, where $W$ is the side length and $d$ is the number of spatial dimensions.

\subsection{A shell model} 
Here, we define the shell model of mixing as in Miles and Doering~\cite{Miles2018}.  Taking the Fourier transform of the advection-diffusion equation \eqref{eq:AD}, we obtain the infinite-dimensional ODE system:
\begin{eqnarray}
&    \partial_t\hat{\theta}(\mathbf{k},t) + \sum_{\mathbf{k'}}i \mathbf{k'}\cdot\hat{\mathbf{u}}(\mathbf{k-k'},t)\hat{\theta}(\mathbf{k'},t) + \kappa|\mathbf{k}|^2\hat{\theta}(\mathbf{k},t) = 0,
\end{eqnarray}
where the wavenumbers $\mathbf{k}$ and $\mathbf{k'}$ belong to $\frac{2\pi}{W} \Z^d$.  Binning the Fourier variables $\hat{\theta}(\mathbf{k},t)$ and $\hat{\mathbf{u}}(\mathbf{k},t)$ with wavenumbers $\sigma^{n-1}k_0 < |\mathbf{k}| < \sigma^n k_0$ into variables $\theta_n(t)$ and $u_n(t)$, respectively, and considering only the mode coupling between neighboring shells, yields the simplest shell model of the form: 
\begin{eqnarray}\label{eq:thetan}
&     \dot{\theta}_n = k_{n-1}u_{n-1}\theta_{n-1} - k_n u_n\theta_{n+1} -\kappa k_n^2 \theta_n,~~~n = 1, 2, 3, \dots,
\end{eqnarray}
where $k_n=\sigma^n k_0$, and $\theta_0 = u_0 \equiv 0$.

One can readily verify that the evolution satisfies the energy identity:
\begin{eqnarray}
    \label{eq:energy_identity}
&\|\theta(t)\|_{L^2}^2 + 2\kappa \int_0^t \|\theta(s)\|_{H^1}^2\, \text{d}s  = \|\theta(0)\|_{L^2}^2,
\end{eqnarray}
where $H^\alpha$ denotes the homogeneous Sobolev (semi)norm, defined for the shell model in Eq.~\eqref{eq:shell_norm} below.

\subsection{Local-in-time and global-in-time optimizations}

In this section, the divergence-free velocity field $\mathbf{u}$ is designed---subject to certain constraints---to maximize either the instantaneous mixing rate (i.e., local-in-time optimization) or the mixing at a prescribed final time $T$ (i.e., global-in-time optimization).  Here we consider the velocity field $\mathbf{u}$ satisfying the \emph{fixed-enstrophy} constraint:
\begin{eqnarray}
&    \int_\Omega |\nabla \times \mathbf{u}|^2 \, \text{d}^d\mathbf{x} = \int_\Omega |\nabla \mathbf{u}|^2 \, \text{d}^d\mathbf{x} = W^d \Gamma^2,
\end{eqnarray}
where $\Gamma$ is the root-mean-square rate-of-strain.

Following Miles and Doering~\cite{Miles2018}, we define the $H^\alpha$ shell-model Sobolev norm as
\begin{equation}
\label{eq:shell_norm}
    \|\psi(t)\|_{H^\alpha}^2 \equiv \sum_{n=1}^{\infty}k_n^{2\alpha} \psi_n^2(t)
\end{equation}
for any vector $\psi = (\psi_1, \psi_2, \psi_3, \dots)$. When $\alpha = -1$, this norm corresponds to the $H^{-1}$ mix-norm $\|\theta(t)\|_{H^{-1}}$, specifically,
\begin{equation}
    \|\theta(t)\|_{H^{-1}}^2 = \sum_{n=1}^{\infty}\frac{\theta_n^2}{k_n^2},
\end{equation}
and when $\alpha = 1$, it provides the expression for enstrophy:
\begin{equation}\label{eq:enstrophy}
    \|\mathbf{u}(t)\|_{H^1}^2 = \sum_{n=1}^{\infty}k_n^2 u_n^2(t) = \Gamma^2.
\end{equation}
Given $u_n$, the values of $\theta_n$ can be obtained by solving Eq.~\eqref{eq:thetan}. In the computations, we use the initial data
\begin{equation}\label{ICs}
    (\theta_1(0),\theta_2(0),\theta_3(0),\dots) = (1,0,0,\dots),
\end{equation}
i.e., starting from the most unmixed state. However, as discussed in the next section, the analytical result may not require such specified initial data \eqref{ICs}.

For local-in-time optimization, the flow field $\mathbf{u}$ realizes
\begin{equation}\label{local-in-time}
    \begin{split}
        &\min_{\mathbf{u}} \frac{\text{d}}{\text{d} t}\|\theta(t)\|_{H^{-1}}^2\\
        &\text{subject to} \quad \|\mathbf{u}\|_{H^1} = \Gamma \;\text{at each time} \; t.
    \end{split}
\end{equation}
A semi-analytical solution for $\mathbf{u}$ is provided in Miles and Doering~\cite{Miles2018} for both energy-constrained and enstrophy-constrained cases.  In this work, we focus on the enstrophy-constrained case.  As informed by our numerical results over an extended time horizon, we demonstrate in the next section that the exponential decay rates of $\theta$ in the last two shells tend to be equal.

For global-in-time optimization at some final time $T$, the flow field $\mathbf{u}$ realizes
\begin{equation}\label{global-in-time}
    \begin{split}
        &\min_{\mathbf{u}} \|\theta(T)\|_{H^{-1}}^2\\
        &\text{subject to} \quad \frac{1}{T}\int_0^T \|\mathbf{u}(t)\|_{H^1}^2 \, \text{d}t = \Gamma^2.
    \end{split}
\end{equation}
Unlike the local-in-time case described in Eq.~\eqref{local-in-time}, where the rate is minimized, here the $H^{-1}$ mix-norm is minimized and the \emph{time-averaged} enstrophy constraint is applied in Eq.~\eqref{global-in-time}. The Lagrange functional corresponding to the optimization problem in Eq.~\eqref{global-in-time} can be expressed as:
\begin{equation}
    \begin{split}
        \mathcal{L}\{\theta,\mathbf{u},\phi,\mu\} = & \frac{1}{2}\sum_{n=1}^{\infty}\frac{\theta_n^2(T)}{k_n^2} + \int_0^T \left\{\sum_{n=1}^{\infty}\phi_n(k_{n-1}u_{n-1}\theta_{n-1}-k_n u_n\theta_{n+1}\right.\\
        &\left.  -\kappa k_n^2\theta_n -\dot{\theta}_n) + \frac{\mu}{2} \left(\sum_{n=1}^{\infty}k_n^2 u_n^2 - \Gamma^2\right) \right\}\, \text{d}t,
    \end{split}
\end{equation}
where $\phi$ is a Lagrange multiplier field enforcing Eq.~\eqref{eq:thetan} and $\mu$ is a scalar `balance' parameter for the enstrophy constraint. At extrema, the first variations (Fr\'{e}chet derivatives) of this functional vanish with respect to $\theta$, $\phi$, $\mathbf{u}$, and $\mu$:
\begin{align}
    \frac{\delta\mathcal{L}}{\delta\theta_n(T)} = 0 &\Rightarrow \frac{\theta_n(T)}{k_n^2} - \phi_n(T) = 0,\label{eq:L1}\\
    \frac{\delta\mathcal{L}}{\delta\theta_n}=0 &\Rightarrow \dot{\phi}_n - k_{n-1}u_{n-1}\phi_{n-1} + k_n u_n\phi_{n+1} - \kappa k_n^2\phi_n = 0,\label{eq:L2}\\
    \frac{\delta\mathcal{L}}{\delta\phi_n}=0 &\Rightarrow \dot{\theta}_n - k_{n-1}u_{n-1}\theta_{n-1} + k_n u_n\theta_{n+1} + \kappa k_n^2\theta_n = 0,\label{eq:L3}\\
    \frac{\delta\mathcal{L}}{\delta u_n}=0 &\Rightarrow k_n\phi_{n+1}\theta_n - k_n\phi_n\theta_{n+1} + \mu k_n^2 u_n = 0,\label{eq:L4}\\
    \frac{\delta\mathcal{L}}{\delta\mu}=0 &\Rightarrow \frac{1}{T}\int_0^T \|\mathbf{u}(t)\|_{H^1}^2 \, \text{d}t - \Gamma^2 = 0.\label{eq:L5}
\end{align}

While Miles and Doering~\cite{Miles2018} have explored the short-time (i.e., small $T$) behavior of optimal mixing for both the local-in-time optimization~\eqref{local-in-time} and global-in-time optimization~\eqref{global-in-time} at $\kappa=0$ and 0.01, this investigation focuses on the long-time (i.e., large $T$) asymptotic behavior of optimal mixing across varying $\kappa$.

\subsection{Numerical methods} 
For local-in-time optimization, the problem in Eq.~\eqref{local-in-time} can be solved numerically by integrating Eq.~\eqref{eq:thetan} based on the semi-analytical solution of $\mathbf{u}$ proposed in Miles and Doering~\citep{Miles2018}.  

For global-in-time optimization, Eqs.~\eqref{eq:L1}--\eqref{eq:L5} were solved numerically in Miles and Doering~\citep{Miles2018} using a gradient-based method (steepest descent in $u_n$ and steepest ascent in $\mu$).   This method has a linear convergence rate and works well for non-stiff systems \citep{Nocedal2006}, such as those with small final time $T$~\citep{Miles2018}.  However, for large $T$, the system becomes stiff, leading to dramatically slower convergence or even failure to converge.  In this study, we solve Eqs.~\eqref{eq:L1}--\eqref{eq:L5} using a time-stepping method (an analogous gradient-based approach) as in Ref.~\cite{Wen2015PRE} to achieve the asymptotic regime of optimal mixing for sufficiently large $T$.

We apply steepest descent in $u_n$ and steepest ascent in $\mu$ by introducing a `pseudo-time' $\tau$ into Eqs.~\eqref{eq:L4} and \eqref{eq:L5}, thereby relaxing the optimality condition and budget constraint, respectively:
 \begin{align}
     \frac{\partial u_n}{\partial \tau} + \frac{\delta\mathcal{L}}{\delta u_n}=0& \Rightarrow \frac{\partial u_n}{\partial \tau} + (k_n\phi_{n+1}\theta_n - k_n\phi_{n}\theta_{n+1}+\mu k_n^2 u_n) =0,\label{eq:L4'}\\
     \frac{\partial \mu}{\partial \tau} -  \frac{\delta\mathcal{L}}{\delta\mu}=0&  \Rightarrow \frac{\partial \mu}{\partial \tau} - \frac{1}{T}\int_0^T \|\mathbf{u}(t)\|_{H^1}^2 \text{d}t + \Gamma^2 = 0.\label{eq:L5'}
\end{align}
These `time-dependent' equations are then advanced until they converge to a stationary solution ($\partial_{\tau} = 0$), which corresponds to the solution of the original Eqs.~\eqref{eq:L4} and \eqref{eq:L5}.  In this work, we update $u_n$ and $\mu$ as follows:
\begin{align}
    u_n^{(s+1)} &= u_n^{(s)}
        - \frac{\Delta\tau_u}{k_n^\eta}
          \left(\frac{\delta\mathcal{L}}{\delta u_n}\right)^{\!(s)},
        \label{eq:L4''}\\
    \mu^{(s+1)} &= \mu^{(s)}
        + \Delta\tau_\mu
          \left(\frac{\delta\mathcal{L}}{\delta\mu}\right)^{\!(s)},
        \label{eq:L5''}
\end{align}
where ${\Delta \tau_u}/{k_n^\eta}$ and $\Delta \tau_{\mu}$ are `time steps' advancing $u_n$ and $\mu$, respectively.  The parameter $\eta>0$ ensures smaller steps at higher wavenumbers, which is crucial for convergence at large $T$.  Initially, we set $\eta = 1$ at the early stage of each computation and increase $\eta$ if necessary after certain iterations to improve convergence.

Computations are performed for a range of $\kappa$ values from $\kappa = 0$ to $\kappa = 0.064$ with $\sigma = \sqrt{2}$, $k_0 = 1/\sqrt{2}$, and $\Gamma = 1$.  In the global-in-time scenario, the numerical algorithm starts with initial guesses $\theta(t) = \theta^{(0)}(t)$, $\mathbf{u}(t) = \mathbf{u}^{(0)}(t)$, and $\mu = \mu^{(0)}$.  At each step $s$, the following procedure is implemented: (i)~Eq.~\eqref{eq:L3} is integrated forward in time to update $\theta(t)$; (ii)~Eq.~\eqref{eq:L2} is integrated backward in time to update $\phi(t)$, using the terminal condition $\phi_n(T) = \theta_n(T)/k_n^2$ as an `initial' condition; and (iii)~$\mathbf{u}(t)$ and $\mu$ are updated via Eqs.~\eqref{eq:L4''} and \eqref{eq:L5''}.
This process continues until the following convergence criteria are met:
\begin{eqnarray}
    \left\|\frac{\delta\mathcal{L}}{\delta u_n}\right\|_\infty < 10^{-11}
    \quad\text{and}\quad
    \left|\frac{\delta\mathcal{L}}{\delta\mu}\right| < 10^{-14}.
\end{eqnarray}

\section{Numerical results and discussion} \label{sec:Results}
In this section, we present numerical results for large $T$ to examine the asymptotic behavior of optimal mixing based on the $H^{-1}$ mix-norm.

\subsection{Local-in-time optimization}
\label{sec:Results_Local}

Consider the local-in-time strategy with fixed enstrophy for the infinite system, starting from the most unmixed state as described in Eq.~\eqref{ICs}.  Figure~\ref{fig:localtheta} illustrates the evolution of the state $\theta_n$ under the optimal control $u_n$ shown in Fig.~\ref{fig:localu}.  In this approach, each component of the control vector $\mathbf{u}$ is applied sequentially and in a piecewise manner over time. The state vector $\theta$ exhibits a distinctive pattern in the local-in-time strategy: only two scales of filaments are present simultaneously; specifically, as the larger scale diminishes, the smaller scale emerges.

\begin{figure}[t]
  \centering
  \includegraphics[width=0.9\textwidth]{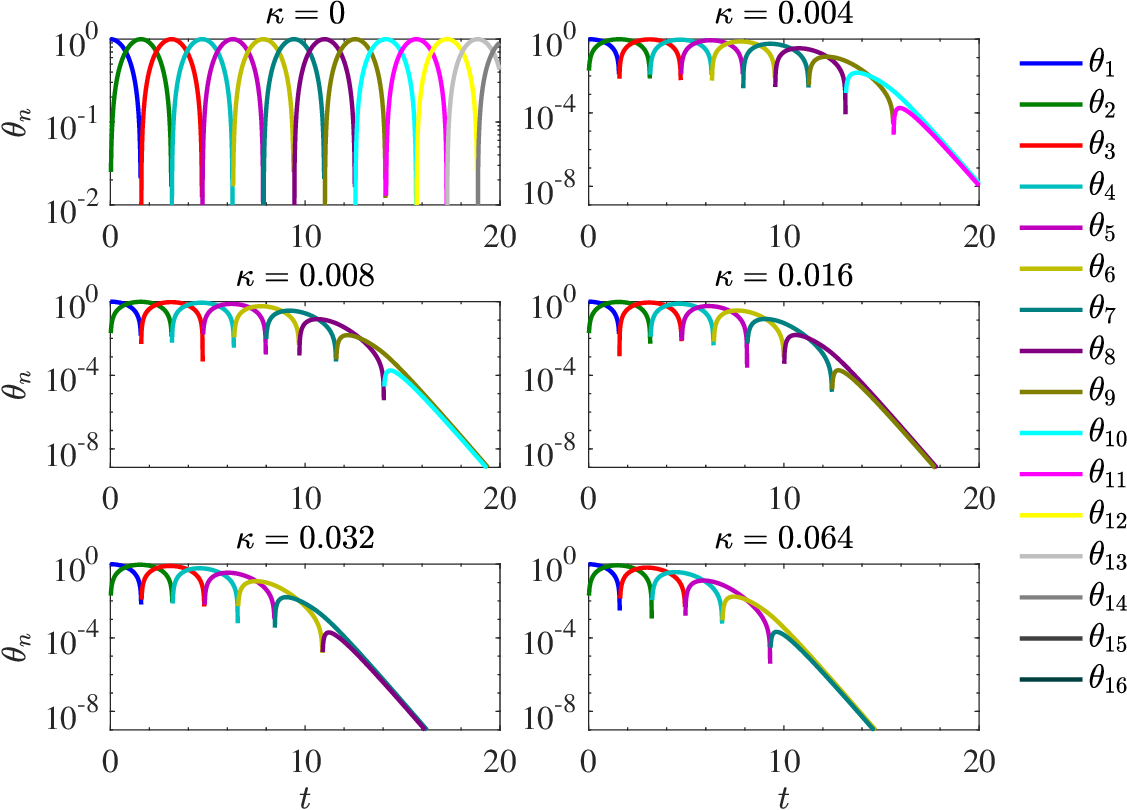}
  \caption{Evolution of the state $\theta_n$ under the local-in-time strategy for different diffusivities $\kappa$.  The mixing process starts from the most unmixed state as described in Eq.~\eqref{ICs}.  For $\kappa = 0$, the state can be described by a sine/cosine function sequentially at each time interval, matching well with the analytical solution~\eqref{eq:soln_localkappa0} given by Miles and Doering~\cite{Miles2018}.  When $\kappa > 0$, the passive scalar spectrum exhausts at a certain small length scale, known as the Batchelor length scale, which increases with $\kappa$.  For $\kappa = 0.004$, 0.008, 0.016, 0.032 and 0.064, the corresponding Batchelor shell numbers $n_b$ are 11, 10, 9, 8 and 7, respectively, aligning perfectly with the prediction by Eq.~\eqref{eq:nb2}.
   }
  \label{fig:localtheta}
\end{figure} 

\begin{figure}[t]
  \centering
  \includegraphics[width=0.9\textwidth]{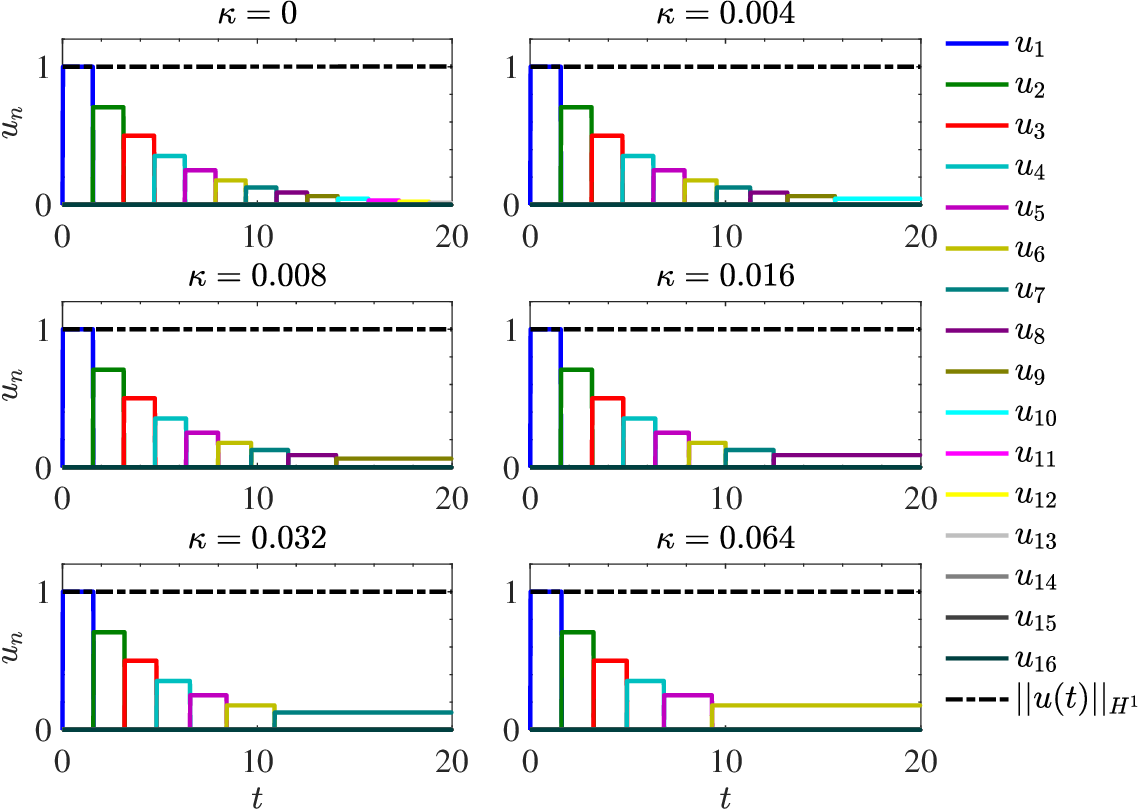}
  \caption{Evolution of the optimal control $u_n$ under the local-in-time strategy for different diffusivities $\kappa$.  Here, $u_n$ is computed following the approach outlined by Miles and Doering~\cite{Miles2018}.  Under the local-in-time strategy~\eqref{local-in-time}, the enstrophy budget $\|\mathbf{u}\|_{H^1}$ remains constant at $\Gamma=1$ throughout the entire process.
   }
  \label{fig:localu}
\end{figure} 

Miles and Doering~\cite{Miles2018} derived a theoretical solution for $\kappa = 0$.  During the interval $t_n \le t \le t_{n+1}$, where $t_n = (n-1)\pi/(2\Gamma)$ marks the time when the state vector is fully within the $n$-th shell (i.e., $\theta_n = 1$ and $\theta_{m\ne n} = 0$), the optimal control and corresponding state can be expressed as: 
\begin{eqnarray}\label{eq:soln_localkappa0}
 &       u_n(t) = \dfrac{\Gamma}{k_n}, \;\; \theta_n = \cos(\Gamma(t-t_n)) \;\; \text{and} \;\; \theta_{n+1} = \sin(\Gamma(t-t_n)),
\end{eqnarray}
with all other components of $\mathbf{u}$ and $\theta$ being zero.
For this solution, the mix-norm at $t = t_n$ is given by 
\begin{eqnarray}\label{eq:mixnorm_localkappa0}
 &       \|\theta(t_n)\|_{H^{-1}} = \|\theta(0)\|_{H^{-1}}e^{-\ln(2)\Gamma t_n/\pi},
\end{eqnarray}
given that their $\sigma = 2$. 
As demonstrated in Figs.~\ref{fig:localtheta} and~\ref{fig:localu}, our numerical results for $\kappa=0$ align well with the theoretical solution \eqref{eq:soln_localkappa0}.

As $t$ increases, the optimal mixing without diffusion (i.e., $\kappa=0$) produces a self-similar cascade of passive scalar fluctuations to increasingly smaller scales.  However, when diffusion is introduced (i.e., $\kappa>0$), the scalar length scale becomes limited by a generalized Batchelor scale, which corresponds to the smallest length scale observed in Fig.~\ref{fig:localtheta} after a long time.  As shown in Fig.~\ref{fig:localtheta}, the mixing eventually involves two adjacent small-scale shells that share the same decay rate as $t\to\infty$. With increasing $\kappa$, the Batchelor scale becomes larger and is reached more quickly (Figs.~\ref{fig:localtheta} and~\ref{fig:localu}).  These findings are consistent with recent studies in both the shell model and PDE contexts~\citep{Miles2018, Lin2011, Miles2018diffusion}.

It has been demonstrated by Miles and Doering~\citep{Miles2018} that for $\kappa > 0$ under the local-in-time strategy, the fixed point of Eq.~\eqref{eq:thetan} at the origin changes from a stable spiral to a stable node as $n$ increases to a critical shell $n_c$~\eqref{eq:nc}.  This transition implies that the state trajectory \emph{cannot} progress to the next plane (i.e., the $\theta_{n_c+1}$--$\theta_{n_c+2}$ plane) since it must intercept the $\theta_{n_c+1}$ axis to do so (see their figure 4).  As a result, the state trajectory is ultimately confined to the $\theta_{n_c}$--$\theta_{n_c+1}$ plane, meaning that the shell number corresponding to the smallest length scales of fluctuations in scalar concentration, i.e., the Batchelor shell $n_b$, can be expressed as $n_b = n_c + 1$.  Based on this transition feature, Miles and Doering~\citep{Miles2018} derived the expression for $n_c$ with $\sigma = 2$.  Below we extend their derivation to an arbitrary $\sigma > 1$.
\begin{theorem}
\label{thm1}
For $\kappa > 0$ under the local-in-time strategy, the fixed point at the origin transitions from a stable spiral to a stable node when
\begin{equation}\label{eq:nc}
    n\geq n_c = \left\lceil\frac{1}{2}\log_\sigma \frac{2\Gamma}{(\sigma^2-1)\kappa k_0^2}\right\rceil.
\end{equation}
\begin{proof}
As noted in Miles and Doering~\cite{Miles2018} and also shown in Fig.~\ref{fig:localu}, on the interval $[t_n, t_{n+1}]$, the optimal control with the enstrophy-constrained local-in-time strategy is $u_n(t) = {\Gamma}/{k_n}$, with the other components of $u$ set to zero, as long as $\theta_n(t)$ and $\theta_{n+1}(t)$ remain positive at time $t$. In the $\theta_n$--$\theta_{n+1}$ plane, the local-in-time strategy is governed piecewise in time by the two-state ODE system:
\begin{equation}\label{eq:thetaODE_local}
    \frac{d}{dt}\begin{pmatrix}
    \theta_{n}\\\theta_{n+1}
    \end{pmatrix}
    = \begin{pmatrix}
    -\kappa k_{n}^2 & -k_{n} u_{n}\\
    k_{n} u_{n}  & -\kappa k_{n+1}^2
\end{pmatrix}
\begin{pmatrix}
    \theta_{n}\\\theta_{n+1}
    \end{pmatrix}.
\end{equation}
Solving this eigenvalue problem yields
\begin{equation}\label{eq:eigen}
    \lambda_{\pm} = -\frac{1}{2}\kappa(k_{n}^2 + k_{n+1}^2) \pm \frac{1}{2}\sqrt{\kappa^2 (k_{n}^2 - k_{n+1}^2)^2 - 4\Gamma^2},
\end{equation}
and the corresponding eigenvectors in the $\theta_{n}$--$\theta_{n+1}$ plane
\begin{equation}
    \Theta_{\pm} = \begin{pmatrix}
        \frac{1}{2}\kappa(k_{n+1}^2 - k_{n}^2) \pm \frac{1}{2}\beta_{n} \\
        \Gamma
    \end{pmatrix},
\end{equation}
where $\beta_n = \sqrt{\kappa^2(k_{n+1}^2-k_{n}^2)^2 - 4\Gamma^2}$.  Therefore, the solution of Eq.~\eqref{eq:thetaODE_local} can be expressed as
\begin{equation}\label{eq:thetasoln}
\begin{split}
 \theta_{n}(t) =  \frac{\theta^i_{n}}{\beta_{n}}\left(\alpha_{+}e^{\lambda_+t} - \alpha_{-}e^{\lambda_-t}\right) \;\;\text{and}\;\; \theta_{n+1}(t) =  \frac{\theta^i_{n}\Gamma}{\beta_{n}}(e^{\lambda_{+}t} - e^{\lambda_{-}t}),
\end{split}
\end{equation}
where $\theta^i_n$ represents the initial value of $\theta_n$ in the $n$-{th} time interval $[t_n, t_{n+1}]$, and $\alpha_{\pm} = \frac{1}{2}\kappa(k_{n+1}^2 - k_{n}^2) \pm \frac{1}{2}\beta_{n}$.

When the imaginary parts of $\lambda_{\pm}$ become zero, the origin, which is one fixed point of Eq.~\eqref{eq:thetaODE_local}, transitions from a stable spiral to a stable node.  At this point, the trajectory can no longer move to the next plane, i.e., the $\theta_{n+1}$--$\theta_{n+2}$ plane.  This occurs when the following inequality is satisfied:
\begin{eqnarray}
  &  \kappa^2(k_{n}^2 - k_{n+1}^2)^2 \geq 4\Gamma^2, 
\end{eqnarray}
indicating that the trajectory will remain within the $\theta_n$--$\theta_{n+1}$ plane.  Simplifying this inequality, we find 
\begin{eqnarray*}
  &  n \geq n_c = \left\lceil\dfrac{1}{2}\log_\sigma \dfrac{2\Gamma}{(\sigma^2-1)\kappa k_0^2}\right\rceil.
\end{eqnarray*}
This completes the proof.
\end{proof}
\end{theorem}

Given Eq.~\eqref{eq:nc}, the Batchelor shell number can be expressed as
\begin{eqnarray}\label{eq:nb}
    &n_b = n_c + 1 = \left\lceil\dfrac{1}{2}\log_\sigma\dfrac{2\Gamma}{(\sigma^2-1)\kappa k_0^2}\right\rceil +1.
\end{eqnarray}
Correspondingly, the Batchelor length scale is given by:
\begin{eqnarray}\label{eq:lb}
   &l_b \sim \dfrac{1}{k_{n_b}} \approx \sqrt{\dfrac{(\sigma^2-1)\kappa}{2\sigma^2\Gamma}}.
\end{eqnarray}
Substituting $\sigma=\sqrt{2}$, $k_0=1/\sqrt{2}$ and $\Gamma=1$ into Eq.~\eqref{eq:nb} yields 
\begin{eqnarray}\label{eq:nb2}
    &n_b = \left\lceil3 - \log_{\sqrt{2}} \sqrt{\kappa}\right\rceil.
\end{eqnarray}
For $\kappa = 0.001\times 2^{m}$ with integer $m$, this simplifies to 
\begin{eqnarray}\label{eq:nb2a}
    &n_b = 13 - m,
\end{eqnarray}
which matches precisely with the numerical data presented in Fig.~\ref{fig:localtheta}.  Furthermore, Eq.~\eqref{eq:thetasoln} indicates that as $t\to\infty$, the final two shells, specifically the $n_c$-th and $(n_{c}+1)$-th shells, will decay exponentially at the same rate $\lambda_+$, consistent with the behavior observed in Fig.~\ref{fig:localtheta}.

\begin{figure}[t]
  \centering
  \includegraphics[width=0.8\textwidth]{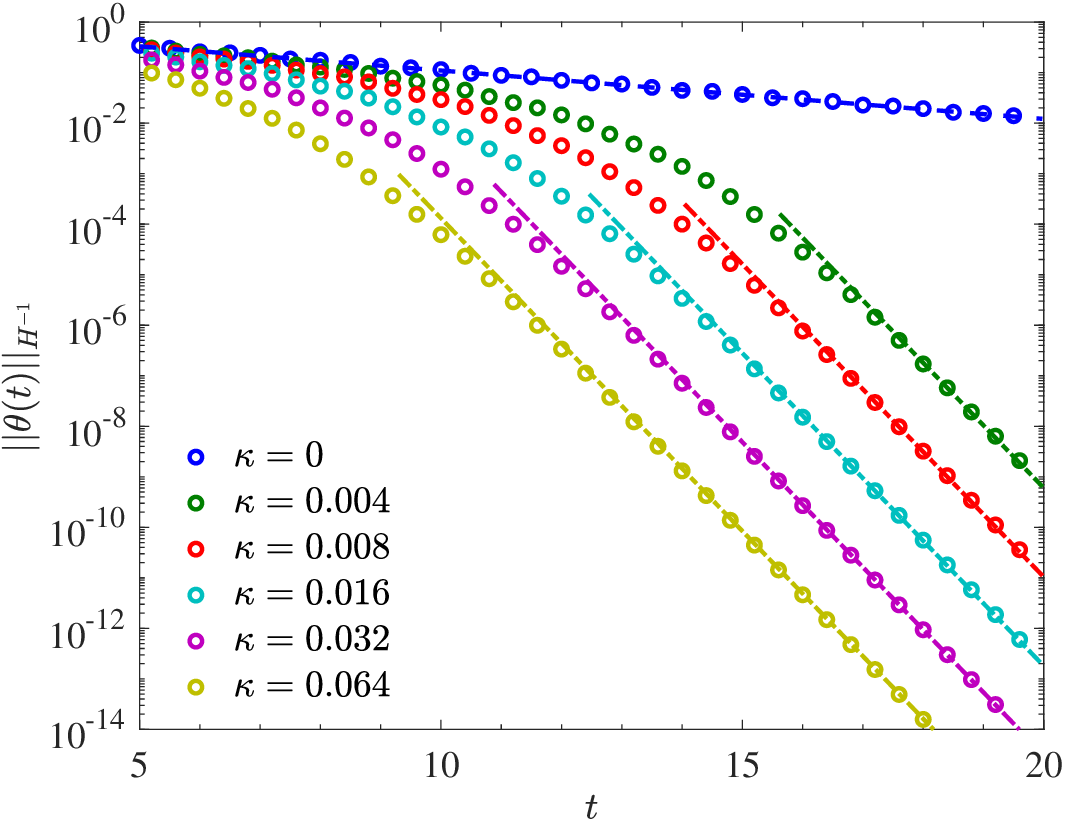}
  \caption{Evolution of the mix-norm $\|\theta(t)\|_{H^{-1}}$ under the local-in-time strategy for different diffusivities $\kappa$.  Circles denote numerical results; the blue dashed line shows the analytical prediction at $\kappa=0$, evaluated at $t_n=n\pi/2$ using~\eqref{eq:mixnorm_localkappa0}; and the dashed-dot lines represent the large-$t$ fits for $\kappa>0$, given by $\exp({-2.852t - 4.056\log_2(\kappa/0.004) + 35.80})$.  The decay rate $-2.852$ in the fitting function is determined from the eigenvalue $\lambda_+$ in~\eqref{eq:eigen}, using the rounded value of $n_b$. 
As $t\rightarrow\infty$, the enstrophy-constrained local-in-time strategy exhibits exponential decay.  For all $\kappa > 0$, the decay follows the same `Batchelor mixing rate', $\|\theta(t)\|_{H^{-1}}\sim \exp({-2.852t - 4.056\log_2(\kappa/0.004) + 35.80})$, while the diffusionless case $\kappa = 0$ displays a distinct decay rate, $\|\theta(t)\|_{H^{-1}}\approx \|\theta(0)\|_{H^{-1}}\exp({-\ln(2)\Gamma t/\pi}) \approx \exp({-0.2206t})$. 
  }
  \label{fig:localnorm}
\end{figure} 

Figure~\ref{fig:localnorm} shows the evolution of the mix-norm $\|\theta(t)\|_{H^{-1}}$ under the local-in-time strategy for different values of diffusivity $\kappa$.  In the long term, both the enstrophy-constrained local-in-time strategies, with and without diffusion, exhibit exponential decay, consistent with recent studies in the PDE context~\citep{Lin2011, Seis2013, Iyer2014, Yao2017, Alberti2019}.  Interestingly, our numerical results highlight distinct exponential decay rates between the strategies with and without diffusion.  For $\kappa=0$,  the numerical data aligns excellently with the analytical prediction, $\|\theta(t)\|_{H^{-1}}\approx \exp(-0.2206t)$, from Eq.~\eqref{eq:mixnorm_localkappa0}.  For $\kappa>0$, the scalar field eventually becomes constrained by the Batchelor length scale, and the mixing is determined by the $(n_b - 1)$-th \& $n_b$-th shells (equivalently, the $n_c$-th \& $(n_c+1)$-th shells).  Consequently, the long-term Batchelor decay rate can be predicted by
\begin{eqnarray}
  & \lambda_b \approx \lambda_{+} = -\dfrac{1}{2}\kappa(k_{n_b - 1}^2 + k_{n_b}^2) + \dfrac{1}{2}\sqrt{\kappa^2 (k_{n_b-1}^2 - k_{n_b}^2)^2 - 4\Gamma^2}. \label{eq:Batchelorrate_local1}
\end{eqnarray}
Substituting the non-rounded $n_b$ into Eq.~\eqref{eq:Batchelorrate_local1} gives
\begin{eqnarray}
   &\lambda_b \approx -\dfrac{\Gamma (\sigma^2 + 1)}{\sigma^2 - 1}, \label{eq:Batchelorrate_local2}
\end{eqnarray}
which is \emph{independent} of $\kappa$, consistent with the numerical results in Fig.~\ref{fig:localnorm} and observations in the PDE context where the long-term rate of mixing under the enstrophy-constrained local-in-time strategy is unaffected by diffusion strength \citep{Miles2018diffusion}.  For the shell model, substituting the rounded $n_b$ and other parameters into Eq.~\eqref{eq:Batchelorrate_local1} yields $\lambda_b = -2.852$.  As shown in Fig.~\ref{fig:localnorm}, for $\kappa>0$, our numerical data shows a strong agreement with this prediction:
\begin{eqnarray}\label{eq:mixnorm_model_b}
   &\|\theta(t)\|_{H^{-1}} \sim \exp\!\left({-2.852t - 4.056\log_2(\kappa/0.004) + 35.80}\right) \;\; \text{as} \;t\to\infty.
\end{eqnarray}
Moreover, our numerical results and analysis suggest that Eq.~\eqref{eq:mixnorm_model_b} holds for any infinitesimal positive $\kappa$, indicating that the enstrophy-constrained local-in-time strategy in the limit $\kappa\to0$ is not equivalent to that at $\kappa=0$.

Our numerical results reveal that for the enstrophy-constrained local-in-time strategy, while diffusion prevents the mixing from generating length scales finer than the Batchelor scale, it significantly enhances the long-term mixing rate compared to the case without diffusion. Additionally, increasing the strength of diffusion further improves mixing effectiveness by reducing the prefactor of the exponential decay.

\subsection{Global-in-time optimization}
\label{sec:Results_Global}

\begin{figure}[t]
  \centering
  \includegraphics[width=0.9\textwidth]{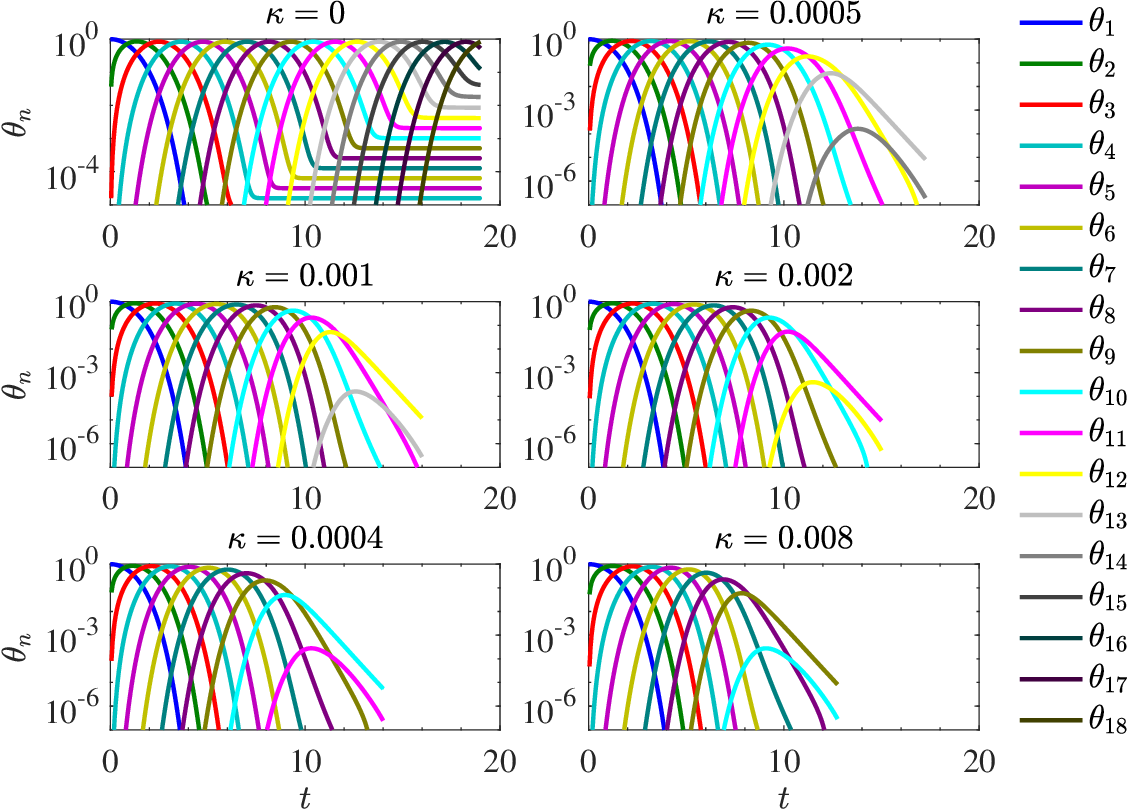}
  \caption{Evolution of the state $\theta_n$ under the global-in-time strategy at large final time $T$ for various diffusivities $\kappa$. Specifically, $T=19$ for $\kappa=0$; $T=17.25$ for $\kappa=0.0005$; $T=16$ for $\kappa=0.001$; $T=15$ for $\kappa=0.002$; $T=14$ for $\kappa=0.004$; and $T=12.75$ for $\kappa=0.008$. The mixing process starts from the maximally unmixed initial condition given in Eq.~\eqref{ICs}.  For $\kappa>0$, as time progresses, the global-in-time strategy converges to a limiting length scale---the Batchelor length scale---which increases with $\kappa$. For $\kappa=0.0005$, 0.001, 0.002, 0.004, and 0.008, the corresponding Batchelor shell numbers $n_b$ are 14, 13, 12, 11, and 10, respectively, in exact agreement with the theoretical predictions of Eq.~\eqref{eq:nb2}.}
  \label{fig:globaltheta}
\end{figure} 

In this subsection, we consider the global-in-time strategy with fixed enstrophy for the infinite system, starting from the most unmixed state~\eqref{ICs}.  A natural question arises: can the global-in-time strategy overcome the limitation imposed by the Batchelor length scale for $\kappa>0$ and thereby produce long-term mixing behavior distinct from that of the local-in-time strategy?

Figures~\ref{fig:globaltheta} and~\ref{fig:globalu} show the evolution of the state $\theta_n$ and the corresponding optimal control $u_n$, obtained by solving Eqs.~\eqref{eq:L1}--\eqref{eq:L5}.  In contrast to the local-in-time strategy, where only two length scales are active at any given time (Fig.~\ref{fig:localtheta}), the global-in-time strategy leads to the emergence of multiple shells as time progresses (Figs.~\ref{fig:globaltheta} and~\ref{fig:globalu}), although many of these shells remain at relatively small amplitudes.  Despite this richer multiscale structure, for $\kappa>0$ the mixing achieved by the global-in-time strategy still appears to be constrained by the Batchelor length scale, in close analogy with the local-in-time case and consistent with the theoretical prediction~\eqref{eq:lb}.  As $\kappa$ increases, the Batchelor scale becomes larger and is reached more rapidly (see Fig.~\ref{fig:globaltheta}). Compared with the local-in-time strategy, the Batchelor shell emerges earlier under the global-in-time strategy, and once it appears, the mixing dynamics are dominated by the Batchelor shell and its immediate neighbor, namely the $n_b$-th and $(n_b-1)$-th shells, as illustrated in Fig.~\ref{fig:globaltheta}.

\begin{figure}[t]
  \centering
  \includegraphics[width=0.9\textwidth]{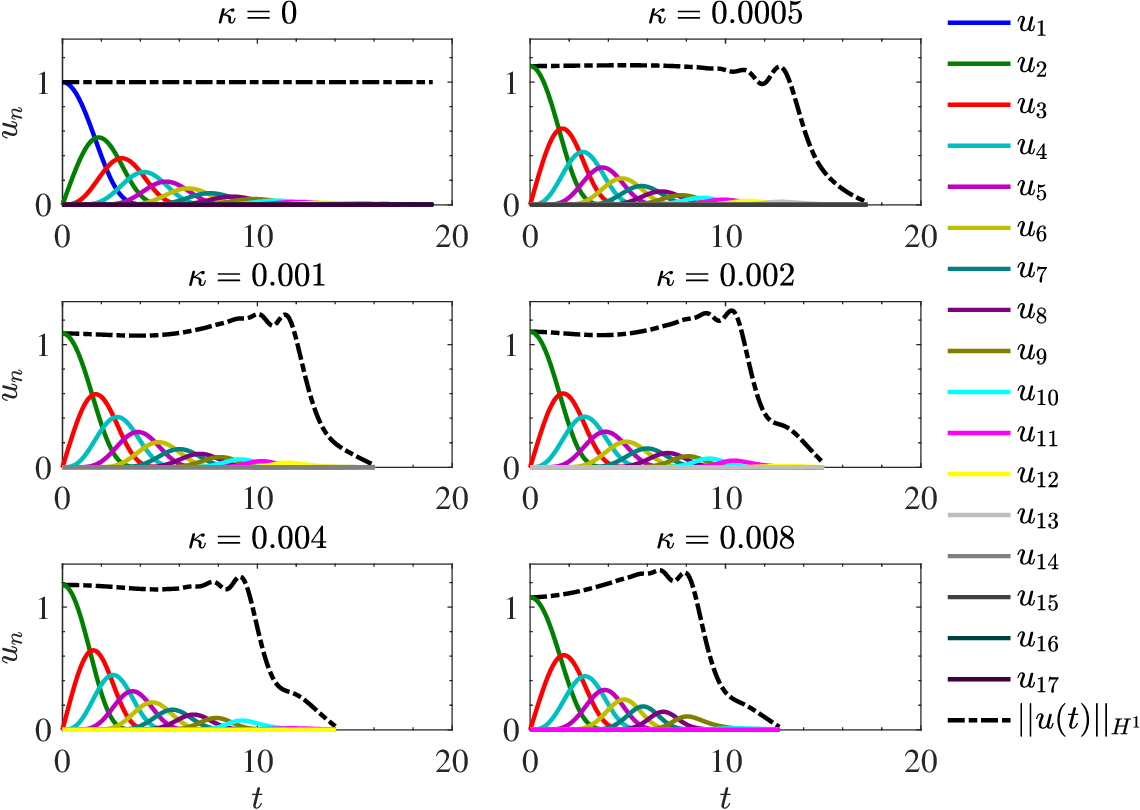}
  \caption{Evolution of the optimal control $u_n$ and the enstrophy budget $\|\mathbf{u}(t)\|_{H^1}$ under the global-in-time strategy at large final time $T$ for various diffusivities $\kappa$. The optimal control $u_n$ corresponds to the state $\theta_n$ shown in Fig.~\ref{fig:globaltheta}.  Prior to reaching the Batchelor length scale, the optimal strategy distributes the enstrophy budget nearly uniformly in time, thereby ensuring efficient control throughout the mixing process.
   }
  \label{fig:globalu}
\end{figure} 

\begin{figure}[t!]
  \centering
  \includegraphics[width=0.8\textwidth]{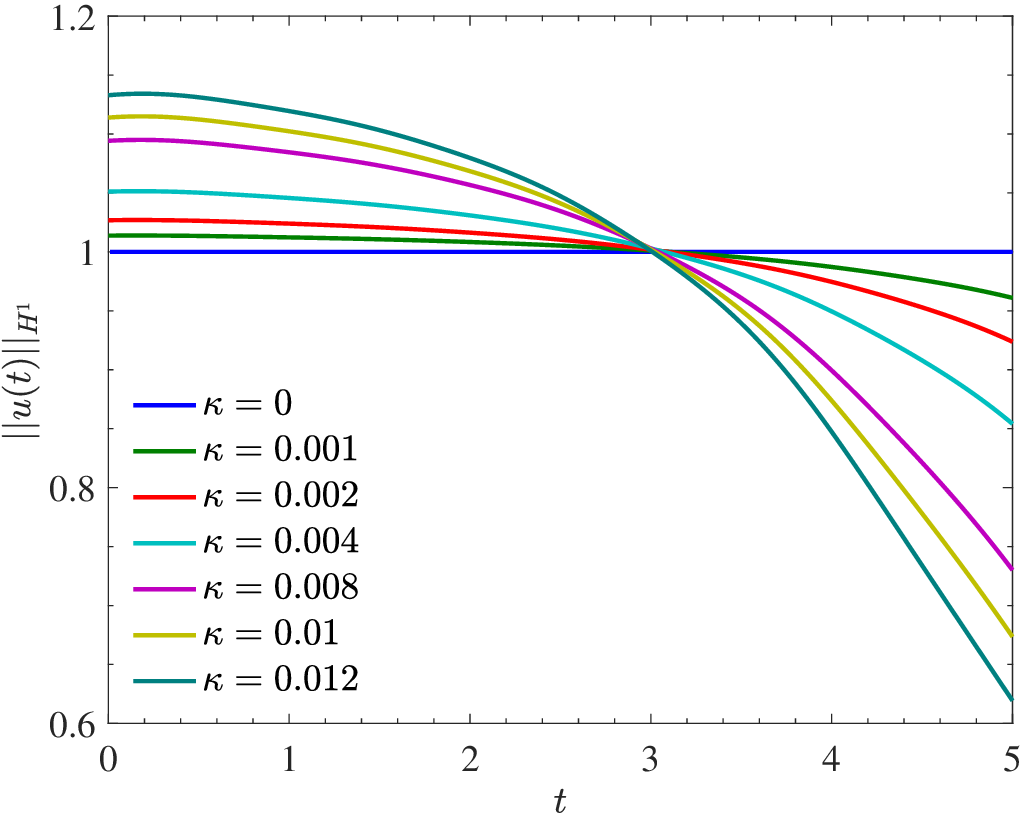}
  \caption{Evolution of the enstrophy budget $\|\mathbf{u}(t)\|_{H^1}$ under the global-in-time strategy with final time $T=5$ for different diffusivities $\kappa$.  As $\kappa$ increases at this short final time, the optimal strategy allocates a larger portion of the enstrophy budget earlier rather than later in the mixing process.}
  \label{fig:budgetT5}
\end{figure} 
\begin{figure}[ht!]
  \centering
  \includegraphics[width=0.8\textwidth]{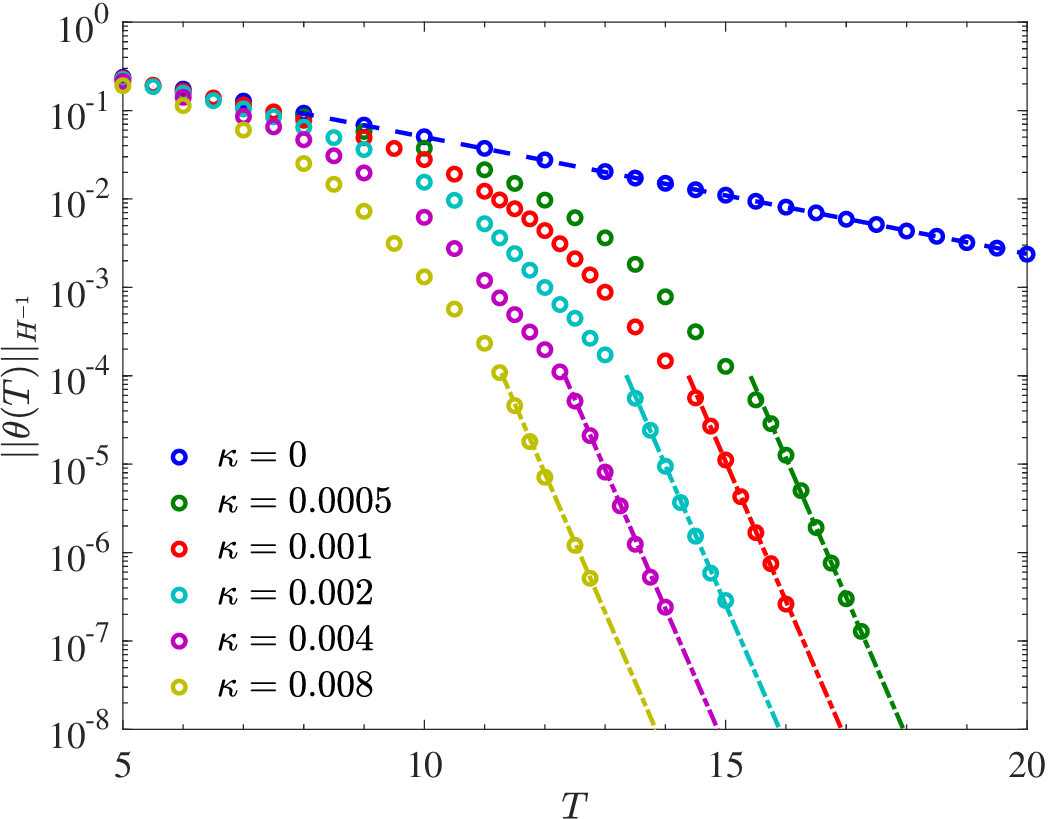}
  \caption{Mix-norm $\|\theta(T)\|_{H^{-1}}$ under the global-in-time strategy as a function of final time $T$ for different diffusivities $\kappa$.  Circles denote numerical results; the blue dashed line shows the large-$T$ fit for $\kappa=0$, given by $\exp({-0.3045T+0.05337})$; and the dashed-dot lines represent the large-$T$ fits for $\kappa>0$, given by $\exp({-3.626t - 3.725\log_2(\kappa/0.004) + 35.49})$.  The mix-norm evolution under the global-in-time strategy exhibits behavior similar to that of the local-in-time strategy shown in Fig.~\ref{fig:localnorm}.
  At large $T$, the enstrophy-constrained global-in-time strategy displays exponential decay.  For $\kappa > 0$, the decay follows the same Batchelor mixing rate, $\|\theta(T)\|_{H^{-1}}\sim \exp({-3.626t - 3.725\log_2(\kappa/0.004) + 35.49})$, whereas for $\kappa = 0$, a distinct decay rate is observed, $\|\theta(T)\|_{H^{-1}}\sim  \exp({-0.3045T+0.05337})$. 
 }
  \label{fig:globalnorms}
\end{figure} 

Below, we discuss how the enstrophy budget $\|\mathbf{u}(t)\|_{H^1}$ is allocated depending on the final time $T$.
For $\kappa = 0$, our numerical results indicate that it is optimal to distribute the enstrophy budget uniformly over time, independent of the magnitude of $T$ (see Figs.~\ref{fig:globalu} and~\ref{fig:budgetT5}). 
This behavior is consistent with previous studies in both the shell-model setting~\citep{Miles2018} and the PDE context~\citep{Mathew2007}---both focused on small-$T$ regimes.
For $\kappa > 0$ and small $T$, both our computations and those of Miles and Doering~\cite{Miles2018} suggest that it is optimal to expend a larger fraction of the enstrophy budget early in time (Fig.~\ref{fig:budgetT5}). This indicates that enstrophy is more effective when applied at larger spatial scales, prior to the influence of the Batchelor length scale. In contrast, for large $T$, our computational results show that, before the Batchelor length scale is reached, it becomes optimal to distribute the enstrophy budget more uniformly over time (Fig.~\ref{fig:globalu}).

Figure~\ref{fig:globalnorms} shows the mix-norm $\|\theta(T)\|_{H^{-1}}$ under the global-in-time strategy as a function of the final time $T$ for different values of $\kappa$.  As in the local-in-time case, the enstrophy-constrained global-in-time strategy---both with and without diffusion---exhibits exponential decay at large $T$, in qualitative agreement with recent studies in the PDE setting \citep{Lin2011, Seis2013, Iyer2014, Yao2017, Alberti2019}.  Moreover, our numerical results indicate that the global-in-time strategies with and without diffusion display distinct exponential decay rates.  For $\kappa = 0$, the mix-norm behaves as $\|\theta(T)\|_{H^{-1}}\sim\exp({-0.3045T+0.05337})$.  In contrast, for $\kappa > 0$, $\|\theta(T)\|_{H^{-1}}\sim \exp({-3.626t - 3.725\log_2(\kappa/0.004) + 35.49})$, with a \emph{$\kappa$-independent} exponential decay rate, similar to that observed for the local-in-time strategy.  Consequently, as in the local-in-time case, the enstrophy-constrained global-in-time strategy in the limit $\kappa \to 0$ is not equivalent to the strategy at $\kappa = 0$.  Notably, for both $\kappa = 0$ and $\kappa > 0$, the global-in-time strategy achieves a strictly larger exponential decay rate than the local-in-time strategy.

\section{Conditional lower bounds}
\label{sec:LowerBound}

Both the analytical and numerical results for the shell model presented in section~\ref{sec:Results_Local} confirm that, under the enstrophy-constrained local-in-time strategy with diffusion, optimal mixing is limited by the Batchelor length scale, with the corresponding shell number given by Eq.~\eqref{eq:nb}. The numerical results in section~\ref{sec:Results_Global} further suggest that the enstrophy-constrained global-in-time strategy is likewise constrained by the same Batchelor scale, although it achieves a higher exponential decay rate.  Below, we derive a lower bound on the mix-norm $\|\theta(t)\|_{H^{-1}}$ for enstrophy-constrained flows. Under Assumption~\ref{assump:Batchelor}, for $\kappa > 0$ we consider a general truncation level $N \geq 1$ such that $u_n(t) = 0$ for all $t \geq 0$ and $n \geq N$, and assume that $\theta_n(0) = 0$ for all $n > N$; the physically relevant case is $N = n_b$, as identified in sections~\ref{sec:Results_Local} and~\ref{sec:Results_Global}.


\begin{assumption}\label{assump:Batchelor}
In the enstrophy-constrained optimal mixing strategy with diffusion
($\kappa > 0$), there exists an integer $N \geq 1$ such that
\[
    u_n(t) = 0 \qquad \text{for all } t \geq 0 \text{ and all } n \geq N.
\]
A particular case, supported by our analytical results and numerical observations in section~\ref{sec:Results}, is $N = n_b$, where
\[
    n_b
    = \left\lceil\frac{1}{2}\log_\sigma\!\left(
      \frac{2\Gamma}{(\sigma^2-1)\kappa k_0^2}
      \right)\right\rceil + 1
\]
is the Batchelor shell number; $n_b - 1$ is then the outermost shell on
which the optimal velocity may be nonzero, and we refer to the $n_b$-th
shell as the \emph{Batchelor shell}.
\end{assumption}

\begin{remark}
If Assumption~\ref{assump:Batchelor} holds, then for the $(N+1)$-th shell we have
\begin{equation}
    \dot{\theta}_{N+1} = k_N u_N \theta_N
    - k_{N+1}u_{N+1}\theta_{N+2}
    - \kappa k_{N+1}^2\theta_{N+1}.
\end{equation}
Since $u_N(t) = 0$ and $u_{N+1}(t) = 0$, there is a unique solution:
\begin{equation*}
    \theta_{N+1}(t) = \theta_{N+1}(0)\,e^{-\kappa k_{N+1}^2 t}.
\end{equation*}
More generally, for all $n > N$, we have $\theta_n(t) = \theta_n(0)\,e^{-\kappa k_n^2 t}$.  If $\theta_n(0) = 0$ for all $n > N$, then $\theta_n(t)$ remains zero at these shells throughout the process; conversely, if the initial condition is non-zero, $\theta_n(t)$ decays exponentially to zero.  In the particular case $N = n_b$, the shells beyond the Batchelor shell are therefore either quiescent or decay exponentially to zero.
\end{remark}

\begin{fact}\label{fact:L2_Hminus1}
Given a sequence $(\theta_1(t),\cdots,\theta_n(t),\cdots)$, if $\theta_n(t) = 0$ for all $n > N$, then
\begin{equation}\label{eq:L2_Hminus1_bound}
    \|\theta(t)\|_{L^2}^2
    \leq k_{N-1}^2\|\theta(t)\|_{H^{-1}}^2
    + \left(1-\frac{1}{\sigma^2}\right)\theta_N^2,
\end{equation}
i.e., the bound \eqref{eq:L2_Hminus1_bound} holds with
$c_1 = k_{N-1}^2$ and $c_2 = 1 - \frac{1}{\sigma^2}$ for any $N \geq 1$.
In the particular case $N = n_b$, we have $c_1 = k_{n_c}^2$, since
$n_b - 1 = n_c$.
\end{fact}

\begin{proof}
By definition of the norms,
\begin{equation*}
    \|\theta(t)\|_{L^2}^2 = \sum_{n=1}^{N}\theta_n^2(t)
    \quad\text{and}\quad
    \|\theta(t)\|_{H^{-1}}^2 = \sum_{n=1}^{N}\frac{\theta_n^2(t)}{k_n^2}.
\end{equation*}
Since $k_n = k_0\sigma^n$ with $\sigma > 1$, we have $k_n \leq k_{N-1}$
for all $n \leq N-1$.
Therefore,
\begin{equation*}
    \begin{split}
        \|\theta(t)\|_{L^2}^2
        &\leq \sum_{n=1}^{N-1}
          \left(\frac{k_{N-1}^2}{k_n^2}\right)\theta_n^2(t)
          + \theta_N^2(t)\\
        &= k_{N-1}^2\sum_{n=1}^{N-1}\frac{\theta_n^2(t)}{k_n^2}
           + \theta_N^2(t)\\
        &= k_{N-1}^2\sum_{n=1}^{N}\frac{\theta_n^2(t)}{k_n^2}
           + \left(1 - \frac{1}{\sigma^2}\right)\theta_N^2(t)\\
        &= k_{N-1}^2\|\theta(t)\|_{H^{-1}}^2
           + \left(1 - \frac{1}{\sigma^2}\right)\theta_N^2(t).
    \end{split}
\end{equation*}
\end{proof}

\begin{remark}\label{rmk:k_n_bound}
Since $n_c = n_b - 1 = \left\lceil\frac{1}{2}\log_\sigma
\frac{2\Gamma}{(\sigma^2-1)\kappa k_0^2}\right\rceil$ is the smallest
integer greater than or equal to the given value, we have
\begin{equation*}
    \frac{1}{2}\log_\sigma\frac{2\Gamma}{(\sigma^2-1)\kappa k_0^2}
    \leq n_c \leq
    \frac{1}{2}\log_\sigma\frac{2\Gamma}{(\sigma^2-1)\kappa k_0^2} + 1.
\end{equation*}
Since $k_{n_c}^2 = k_0^2\sigma^{2n_c}$, we obtain
\begin{equation}\label{eq:knc}
    \frac{2\Gamma}{(\sigma^2-1)\kappa}
    \leq k_{n_c}^2 \leq
    \frac{2\Gamma\sigma^2}{(\sigma^2-1)\kappa}.
\end{equation}
Since $k_{n_b}^2 = \sigma^2 k_{n_c}^2$, multiplying \eqref{eq:knc} through by $\sigma^2$ gives
\begin{equation}\label{eq:knb}
    \frac{2\Gamma\sigma^2}{(\sigma^2-1)\kappa}
    \leq k_{n_b}^2 \leq
    \frac{2\Gamma\sigma^4}{(\sigma^2-1)\kappa}.
\end{equation}
\end{remark}

\begin{theorem}[Lower bounds for the $H^{-1}$ norm]
\label{thm:Hminus1_lower_bound}
Let $k_n = k_0\sigma^n$ with $\sigma > 1$, and assume that
$\|\mathbf{u}(t)\|_{H^1} = \Gamma$ for all times under consideration.
\begin{enumerate}
\item If $\kappa = 0$, then
\begin{equation}\label{eq:lower_kappa0}
    \|\theta(t)\|_{H^{-1}}
    \geq
    \|\theta(0)\|_{H^{-1}}
    \exp\!\left[\Gamma\!\left(\frac{1}{\sigma}-\sigma\right)t\right].
\end{equation}

\item Suppose that $\kappa > 0$ and that, for some $N \geq 1$,
$\theta_n(t) = 0$ for all $n > N$.
Then
\begin{equation}\label{eq:lower_kappaN}
    \|\theta(t)\|_{H^{-1}}
    \geq
    \|\theta(0)\|_{H^{-1}}
    \exp\!\left[\left(
    \Gamma\!\left(\frac{1}{\sigma}-\sigma\right)
    - \kappa k_N^2
    \right)t\right].
\end{equation}

\item Suppose additionally that $N = n_b = n_c + 1$, and that there exist
$T \geq 0$ and a constant $q_T \geq 0$ such that
\[
    \theta_{n_b}^2(t) \leq q_T\theta_{n_c}^2(t), \qquad t \geq T.
\]
Then, for every $t \geq T$,
\[
    \|\theta(t)\|_{H^{-1}}
    \geq
    \|\theta(T)\|_{H^{-1}}
    \exp\!\left[\left(
    \Gamma\!\left(\frac{1}{\sigma}-\sigma\right)
    - \kappa k_{n_c}^2
    \frac{\sigma^2(1+q_T)}{\sigma^2+q_T}
    \right)(t-T)\right].
\]
Moreover, if $\displaystyle\lim_{T\to\infty}q_T = 0$, then for every
$t \geq T$,
\[
    \|\theta(t)\|_{H^{-1}}
    \geq
    \|\theta(T)\|_{H^{-1}}
    \exp\!\left[\left(
    \Gamma\!\left(\frac{1}{\sigma}-\sigma\right)
    - \kappa k_{n_c}^2 - r_T
    \right)(t-T)\right],
\]
where
\[
    r_T := \kappa k_{n_c}^2
    \frac{(\sigma^2-1)q_T}{\sigma^2+q_T}
    \longrightarrow 0
    \quad \text{as } T \to \infty.
\]
In particular, applying \eqref{eq:knc},
\begin{equation}\label{eq:lower_kappa_asym}
    \|\theta(t)\|_{H^{-1}}
    \geq
    \|\theta(T)\|_{H^{-1}}
    \exp\!\left[\left(
    \Gamma\!\left(\frac{1}{\sigma}-\sigma\right)
    - \frac{2\Gamma\sigma^2}{\sigma^2-1} + o_T(1)
    \right)(t-T)\right],
\end{equation}
where $o_T(1) \to 0$ as $T \to \infty$ and is nonpositive.
\end{enumerate}
\end{theorem}

\begin{proof}
Set
\[
    X(t) := \|\theta(t)\|_{H^{-1}}^2.
\]
Using the shell equations and the relation $k_{n+1} = \sigma k_n$,
we obtain
\begin{align}
    \frac{1}{2}X'(t)
    &= \left(\frac{1}{\sigma^2}-1\right)
       \sum_{n}\frac{u_n\theta_n\theta_{n+1}}{k_n}
       - \kappa\|\theta(t)\|_{L^2}^2 \nonumber\\
    &= \left(\frac{1}{\sigma}-\sigma\right)
       \sum_n
       \left(\frac{\theta_n}{k_n}\right)
       \!\left(\frac{\theta_{n+1}}{k_{n+1}}\right)
       k_nu_n
       - \kappa\|\theta(t)\|_{L^2}^2.
    \label{eq:Hminus_identity}
\end{align}
Define
\[
    S(t) := \sum_n
    \left(\frac{\theta_n}{k_n}\right)
    \!\left(\frac{\theta_{n+1}}{k_{n+1}}\right)
    k_nu_n.
\]
By the Cauchy--Schwarz inequality,
\begin{align*}
    |S(t)|
    &\leq
    \left(\sum_n\frac{\theta_n^2}{k_n^2}\right)^{\!1/2}
    \!\!\left(\sum_n\frac{\theta_{n+1}^2}{k_{n+1}^2}k_n^2u_n^2
    \right)^{\!1/2}\\
    &\leq
    \|\theta(t)\|_{H^{-1}}
    \left(\sup_n\frac{|\theta_{n+1}(t)|}{k_{n+1}}\right)
    \!\left(\sum_n k_n^2u_n^2(t)\right)^{\!1/2}\\
    &\leq \Gamma\|\theta(t)\|_{H^{-1}}^2.
\end{align*}
Since $\frac{1}{\sigma} - \sigma < 0$, it follows that
\[
    \left(\frac{1}{\sigma}-\sigma\right)S(t)
    \geq
    \Gamma\!\left(\frac{1}{\sigma}-\sigma\right)\|\theta(t)\|_{H^{-1}}^2.
\]
Therefore, \eqref{eq:Hminus_identity} implies
\begin{equation}\label{eq:basic_Hminus_ineq}
    \frac{1}{2}X'(t)
    \geq
    \Gamma\!\left(\frac{1}{\sigma}-\sigma\right)X(t)
    - \kappa\|\theta(t)\|_{L^2}^2.
\end{equation}

If $\kappa = 0$, integrating \eqref{eq:basic_Hminus_ineq} and taking
square roots immediately gives \eqref{eq:lower_kappa0}.

Now suppose $\kappa > 0$ and $\theta_n(t) = 0$ for all $n > N$.
Then
\[
    \|\theta(t)\|_{L^2}^2
    = \sum_{n=1}^{N}\theta_n^2(t)
    \leq k_N^2\sum_{n=1}^{N}\frac{\theta_n^2(t)}{k_n^2}
    = k_N^2X(t).
\]
Substituting into \eqref{eq:basic_Hminus_ineq} and integrating gives
\eqref{eq:lower_kappaN}.

Finally, suppose $N = n_b = n_c + 1$ and
$\theta_{n_b}^2(t) \leq q_T\theta_{n_c}^2(t)$ for $t \geq T$.
If $q_T = 0$, then $\theta_{n_b}(t) = 0$ for $t \geq T$, and the
estimate reduces to \eqref{eq:lower_kappaN} with $N = n_b - 1 = n_c$.
Assume $q_T > 0$.
By the definition of the $H^{-1}$ norm,
\begin{align*}
    X(t)
    &\geq \frac{\theta_{n_c}^2(t)}{k_{n_c}^2}
           + \frac{\theta_{n_b}^2(t)}{k_{n_b}^2}
    \geq \frac{\theta_{n_b}^2(t)}{q_Tk_{n_c}^2}
           + \frac{\theta_{n_b}^2(t)}{k_{n_b}^2}.
\end{align*}
Since $k_{n_b} = \sigma k_{n_c}$, we have
$\frac{1}{q_Tk_{n_c}^2} = \frac{\sigma^2}{q_Tk_{n_b}^2}$,
and hence
\begin{equation}\label{eq:highest_shell_control}
    \theta_{n_b}^2(t)
    \leq \frac{q_Tk_{n_b}^2}{\sigma^2+q_T}\,X(t).
\end{equation}
On the other hand, by Fact~\ref{fact:L2_Hminus1},
\[
    \|\theta(t)\|_{L^2}^2
    \leq k_{n_c}^2X(t)
    + \left(1-\frac{1}{\sigma^2}\right)\theta_{n_b}^2(t).
\]
Substituting \eqref{eq:highest_shell_control} and using
$k_{n_b}^2 = \sigma^2k_{n_c}^2$, we obtain
\begin{align*}
    \|\theta(t)\|_{L^2}^2
    &\leq \left[k_{n_c}^2
      + \left(1-\frac{1}{\sigma^2}\right)
      \frac{q_Tk_{n_b}^2}{\sigma^2+q_T}\right]X(t)\\
    &= k_{n_c}^2\left[1+\frac{(\sigma^2-1)q_T}{\sigma^2+q_T}\right]X(t)\\
    &= k_{n_c}^2\frac{\sigma^2(1+q_T)}{\sigma^2+q_T}\,X(t).
\end{align*}
Substituting into \eqref{eq:basic_Hminus_ineq} for $t \geq T$,
integrating from $T$ to $t$, and taking square roots yields
\begin{align*}
    \|\theta(t)\|_{H^{-1}}
    &\geq
    \|\theta(T)\|_{H^{-1}}
    \exp\!\left[\left(
    \Gamma\!\left(\frac{1}{\sigma}-\sigma\right)
    - \kappa k_{n_c}^2\frac{\sigma^2(1+q_T)}{\sigma^2+q_T}
    \right)(t-T)\right]\\
    &=
    \|\theta(T)\|_{H^{-1}}
    \exp\!\left[\left(
    \Gamma\!\left(\frac{1}{\sigma}-\sigma\right)
    - \kappa k_{n_c}^2 - r_T
    \right)(t-T)\right]\\
    &\geq
    \|\theta(T)\|_{H^{-1}}
    \exp\!\left[\left(
    \Gamma\!\left(\frac{1}{\sigma}-\sigma\right)
    - \frac{2\Gamma\sigma^2}{\sigma^2-1} + o_T(1)
    \right)(t-T)\right],
\end{align*}
where the last inequality applies the upper bound
$\kappa k_{n_c}^2 \leq \frac{2\Gamma\sigma^2}{\sigma^2-1}$
from~\eqref{eq:knc}, and $o_T(1) = -r_T \to 0$ as $T\to\infty$.
\end{proof}

At $\kappa = 0$, we recover the same lower bound \eqref{eq:lower_kappa0} as that obtained by Miles and Doering~\citep{Miles2018}.  When $\kappa > 0$, the numerical results in Fig.~\ref{fig:globaltheta} support the asymptotic highest-shell depletion assumption $\displaystyle\lim_{T\to\infty}q_T = 0$.  We emphasize that this vanishing of $q_T$ has been verified numerically only for the global-in-time optimization.  For the local-in-time strategy, one instead finds that $q_T \to 1$ as $T \to \infty$.  Since the local-in-time strategy is not expected to achieve faster mixing than the globally optimized strategy, the same lower bound should remain valid in the local-in-time setting; however, the present argument applies only to the global-in-time problem.  It then follows from \eqref{eq:lower_kappa_asym} that as $T \to \infty$, the lower bound
\begin{equation}\label{eq:lower_kappa_final}
    \|\theta(t)\|_{H^{-1}} \gtrsim
    \|\theta(T)\|_{H^{-1}}
    \exp\!\left[\left(
    \Gamma\!\left(\frac{1}{\sigma}-\sigma\right)
    - \frac{2\Gamma\sigma^2}{\sigma^2-1}
    \right)(t-T)\right]
\end{equation}
holds with a decay-rate coefficient that is strictly independent of $\kappa$.

Table~\ref{tab:Exponents} compares the exponential decay rates in the lower bounds derived by Miles and Doering~\citep{Miles2018} and in the present study for different values of $\sigma$ and $\kappa > 0$.  As shown in Eq.~\eqref{eq:lower_kappa_final} and Table~\ref{tab:Exponents}, our lower bound is strictly independent of $\kappa$ for any $\kappa > 0$, consistent with the numerical results presented in Figs.~\ref{fig:localnorm} and~\ref{fig:globalnorms}.  In contrast, the lower bound of Miles and Doering~\citep{Miles2018} retains an explicit $\kappa$-dependence for sufficiently large $\sigma$.  We note, however, that our bound is not uniformly sharper: as shown in Table~\ref{tab:Exponents}($b$), their bound is slightly larger than ours at $\sigma = \sqrt{2}$ and $\Gamma = 1$.

For the parameter values used in our numerical computations ($\sigma = \sqrt{2}$, $\Gamma = 1$), the conditional lower bounds read $\|\theta\|_{H^{-1}} \gtrsim e^{-t/\sqrt{2}} \approx e^{-0.7071t}$
for $\kappa = 0$, and $\|\theta\|_{H^{-1}} \gtrsim e^{-(1/\sqrt{2}+4)t} \approx e^{-4.7071t}$ for $\kappa > 0$.  The numerical results reported in section~\ref{sec:Results}---i.e., $\|\theta\|_{H^{-1}} \sim e^{-0.3045T}$ for $\kappa = 0$ and $\|\theta\|_{H^{-1}} \sim e^{-3.626T}$ for $\kappa > 0$, obtained under the global-in-time strategy---are fully consistent with both the bounds derived in the present study and those of Miles and Doering~\citep{Miles2018}.
 
 \begin{table}[t!]
\centering
\begin{subtable}[t]{0.96\textwidth}
\centering
\caption*{\normalsize($a$) $\sigma = 2,~\Gamma = 1$}
\label{tab:sigma2Gamma1}
     \begin{tabular}{|c|c|c|c|}
      \hline 
      $\kappa$ & $n_b$  & Miles and Doering~\citep{Miles2018} & Present\\
      \hline 
       $0.0005$ & $8$ & $-8.2328$  & $-4.1667$  \\
        \hline 
       $0.001$ & $7$ & $-4.1779$  & $-4.1667$\\
        \hline 
        $0.002$ & $7$ & $-8.2328$  & $-4.1667$\\
         \hline 
         $0.004$ & $6$ & $-4.1779$  & $-4.1667$\\
          \hline 
          $0.008$ & $6$
          & $-8.2328$  & $-4.1667$\\
           \hline 
     \end{tabular}
\end{subtable}
        \hfill
\begin{subtable}[t]{0.96\textwidth}
\centering
\caption*{\normalsize($b$) $\sigma = \sqrt{2},~\Gamma = 1$}
\label{tab:sigmaroot2Gamma1}
     \begin{tabular}{|c|c|c|c|}
      \hline 
      $\kappa$ & $n_b$  &  Miles and Doering~\citep{Miles2018} & Present\\
      \hline 
       $0.0005$ & $14$ & $-4.3567$  & $-4.7071$  \\
        \hline 
       $0.001$ & $13$ & $-4.3567$  & $-4.7071$ \\
        \hline 
        $0.002$ & $12$ & $-4.3567$  & $-4.7071$  \\
         \hline 
         $0.004$ & $11$ &$-4.3567$  & $-4.7071$ \\
          \hline 
          $0.008$ & $10$ &
         $-4.3567$  & $-4.7071$ \\
           \hline 
     \end{tabular}
\end{subtable}
\caption{Comparison of the exponential decay rate in the lower bound between Miles and Doering~\citep{Miles2018} and the present study.
}
\label{tab:Exponents}
\end{table}

\section{Conditional upper bounds}
\label{sec:UpperBound}

In this section, we derive conditional bounds on the maximal rates of enhanced dissipation, complementing the conditional lower bounds on the mix-norm derived in Theorem~\ref{thm:Hminus1_lower_bound}.

We first derive an estimate on the propagation of regularity for the shell model. 
\begin{lemma}
\label{L1}
It holds that
\[
\|\theta(t)\|_{H^1}^2 + 2\kappa \int_0^t \|\theta(s)\|_{H^2}^2\, \emph{d} s \le \exp\left(2(\sigma-\sigma^{-1})\int_0^t \|\u\|_{H^1}\, \emph{d}s\right) \|\theta_0\|_{H^1}^2
\]
for any $t\ge 0$.
\end{lemma}

The proof is actually an adaption of the bound on the $H^{-1}$ norm found in Miles and Doering~\citep{Miles2018}.
\begin{proof}
We compute the rate of change of the (homogeneous) $H^1$ norm under the shell model evolution \eqref{eq:thetan},
\begin{align*}
\frac12\frac{\text{d}}{\text{d}t} \|\theta\|_{H^1}^2 & = \sum_{n\ge 1} k_n^2 \theta_n\dot \theta_n\\
& = \sum_{n\ge1} k_n^2 \theta_n\left(k_{n-1}u_{n-1} \theta_{n-1} - k_n u_n \theta_{n+1}\right) - \kappa \sum_{n\ge1} k_n^4 \theta_n^2.
\end{align*}
We use $\theta_0=0$ and the definition $k_n = \sigma^n k_0$ to rewrite the advection term,
\begin{align*}
\frac12\frac{\text{d}}{\text{d}t} \|\theta\|_{H^1}^2 + \kappa \|\theta\|_{H^2}^2  & =  \sum_{n\ge1} k_n \left(k_{n+1}^2-k_n^2\right)\theta_n \theta_{n+1} k_n u_n\\
& = \left(\sigma -\sigma^{-1}\right) \sum_{n\ge 1} k_nk_{n+1}\theta_n\theta_{n+1}k_n u_n.
\end{align*}
With the help of the Cauchy--Schwarz inequality and the trivial mode-wise bound $k_n |u_n|\le \|\u\|_{H^1}$, we furthermore estimate
\begin{align*}
\frac12\frac{\text{d}}{\text{d}t} \|\theta\|_{H^1}^2 + \kappa \|\theta\|_{H^2}^2  \le \left(\sigma -\sigma^{-1}\right) \| \u\|_{H^1} \|\theta\|_{H^1}^2.
\end{align*}
We invoke a Gronwall-type argument to deduce the statement of the lemma.
\end{proof}

In the next step, we convert the bound on the regularity into an estimate on the dissipation.

\begin{lemma}
\label{L2}
It holds that
\[
\kappa \int_0^t \|\theta\|_{H^1}^2\, \emph{d}t\le \sqrt{\frac{\kappa t}2} \exp\left((\sigma-\sigma^{-1}) \int_0^t \|\u\|_{H^1}\, \emph{d}t\right)\|\theta_0\|_{H^1}\|\theta_0\|_{L^2},
\]
for any $t\ge 0$.
\end{lemma}

\begin{proof}
The assertion is an immediate consequence of the previous regularity lemma and an interpolation. Indeed, since $\|\theta\|_{H^1}^2\le \|\theta\|_{L^2}\|\theta\|_{H^2}$ as a consequence of the Cauchy--Schwarz inequality, we deduce from the energy balance \eqref{eq:energy_identity} that
\begin{align*}
    \kappa \int_0^t \|\theta\|_{H^1}^2\, \text{d}t & \le \kappa \|\theta_0 \|_{L^2}\int_0^t \|\theta\|_{H^2}\, \text{d}t.
\end{align*}
We now invoke Jensen's inequality in the time integral and Lemma \ref{L1} to estimate
\begin{align*}
     \kappa \int_0^t \|\theta\|_{H^1}^2\, \text{d}t & \le \sqrt{\kappa t} \|\theta_0 \|_{L^2} \left(\kappa \int_0^t \|\theta\|_{H^2}^2\, \text{d}t\right)^{1/2}\\
     &\le  \sqrt{\frac{\kappa t}2} \|\theta_0 \|_{L^2} \exp\left((\sigma-\sigma^{-1})\int_0^t \|\u\|_{H^1}\, \text{d}s\right) \|\theta_0\|_{H^1}.
\end{align*}
\end{proof}

The following estimate gives  a bound on the maximal rate of enhanced dissipation that is possible in the  diffusive shell model under the assumption that the vector field $\u$ has a fixed enstrophy budget, see Eq.~\eqref{eq:enstrophy}.
\begin{theorem}
\label{T1}
Let $D>0$ be such that
\begin{equation}\label{2}
\|\theta(t)\|_{L^2} \le e^{-D t}\|\theta_0\|_{L^2},
\end{equation}
for any $t\ge0$, and let $\gamma_0=\|\theta_0\|_{H^1}/\|\theta_0\|_{L^2}$ denote the initial frequency scale. Then
\begin{equation}\label{3}
D \le 2\frac{(\sigma-\sigma^{-1})\Gamma}{\log\frac1{\kappa}},
\end{equation}
for any $\kappa\le \kappa_0$  for some sufficiently small $\kappa_0$ dependent on $\gamma_0$, $\sigma$ and $\Gamma$.
\end{theorem}

\begin{proof}
Our starting point is the energy balance \eqref{eq:energy_identity} in which we make use of the hypothesis \eqref{2} and the bound on the dissipation in Lemma \ref{L2} to the effect that
\[
\|\theta_0\|_{L^2}^2 \le e^{-2 Dt } \|\theta_0\|_{L^2} ^2+ \sqrt{2\kappa t} \exp\left((\sigma-\sigma^{-1}) \int_0^t \|\u\|_{H^1}\, \text{d}t\right)\|\theta_0\|_{L^2}\|\theta_0\|_{H^1}.
\]
Invoking the enstrophy budget assumption in \eqref{eq:enstrophy}, the latter can be rewritten as
\[
1 \le e^{-2Dt} + {\sqrt{2\kappa t}}{\gamma_0} \exp\left((\sigma-\sigma^{-1})t \Gamma\right),
\]
which holds true for any $t\ge 0$. The best bound on $D$ can be obtained by minimizing the right-hand side with respect to $t$. Since we are interested into the scaling behavior only, we are satisfied with giving a  rough bound by picking $t$ suitably. For instance, if we choose $t$ such that the first term on the right-hand side equals $1/2$, i.e.,
\[
t = \frac{\log 2}{2D},
\]
which is the enhanced dissipation time scale, we obtain that
\[
\frac1{\kappa}  \le {4\log 2}{\gamma_0^2} \frac1D  \exp\left(\log( 2) (\sigma-\sigma^{-1})\frac{\Gamma}D\right). 
\]
Applying the logarithm,  we thus have that
\[
\log \frac1{\kappa} \le \log \left({4\log2}{\gamma_0^2}\right) + \log \frac1D + \log( 2) (\sigma-\sigma^{-1}) \frac{\Gamma}D.
\]
If $\kappa $ is chosen small enough, say
\[
\kappa \le \frac1{16 (\log 2)^2\gamma_0^4},
\]
we can simplify the bound as
\[
\log \frac1{\sqrt{\kappa}} \le \log\frac1{D} + \log(2)(\sigma-\sigma^{-1})\frac{\Gamma}D.
\]
Because the left-hand side is diverging if $\kappa$ is getting small, also the dissipation rate $D=D(\kappa)$ has to decrease if $\kappa $ is getting small, $D(\kappa)\to 0$ as $\kappa\to0$. Therefore, the dominant term on the right-hand side is the second one, and we have that
\[
\log \frac1{\kappa} \le 2 (\sigma-\sigma^{-1})\frac{\Gamma}D
\] 
if $\kappa\le \kappa_0$ for some isufficiently small $\kappa_0$ dependent on $\gamma_0,\sigma$ and $\Gamma$. This proves the statement of the theorem.
\end{proof}

The postulated bound on the rate of enhanced dissipation \eqref{2} supposedly captures the optimal rate of decay of the $L^2$ norm in the early states of the mixing process as it is sharp for $t=0$. For later times, the exponential rate should be determined by the Batchelor wave number $O(1/\sqrt{\kappa})$, so that $D$ should become independently of $\kappa$. Instead, motivated by the numerical studies, there should be a diverging prefactor multiplying the exponential function. The intention of our second theorem is to give a rigorous bound on this prefactor.

\begin{theorem}
\label{T2}
Let $D,\Lambda>0$ be such that
\begin{equation}\label{5}
\|\theta(t)\|_{L^2} \le \Lambda e^{-D t}\|\theta_0\|_{L^2},
\end{equation}
for any $t\ge0$,    and let $\gamma_0=\|\theta_0\|_{H^1}/\|\theta_0\|_{L^2}$ denote the frequency length scale, and assume that $D$ is independent of $\kappa$. Then
\begin{equation}\label{6}
\Lambda \ge \frac1{\kappa^{\alpha}}
\end{equation}
for some $\alpha\in(0,1)$ and any $\kappa\le \kappa_0$, where $\alpha$ and $\kappa_0$ both dependend on $\gamma_0,\Gamma , \sigma$ and $D$.
\end{theorem}

\begin{proof}
The proof is very similar to that of the previous theorem. This time, we derive
\[
1\le \Lambda e^{-2Dt} + {\sqrt{2\kappa t}}{\gamma_0} \exp\left((\sigma-\sigma^{-1})t\Gamma \right)
\]
from the energy balance  \eqref{eq:energy_identity}, the dissipation  hypothesis \eqref{5}, the bound in Lemma \ref{L2} and the enstrophy budget assumption in Eq.~\eqref{eq:enstrophy}. Here, $\gamma_0$ is, as before, the frequency scale of the initial configuration. For
\[
t = \frac{\log(2\Lambda)}{2D},
\]
the first term on the right-hand side equals $1/2$, so that
\[
\frac1{\kappa} \le \frac{4\log(2\Lambda)\gamma_0^2 }{D} \exp\left((\sigma-\sigma^{-1}) \frac{\log(2\Lambda)\Gamma}{D}\right).
\]
We apply the logarithm to this inequality and find
\[
\log\frac1{\kappa} \le \log\frac{4\gamma_0^2 }{D} + \log\log (2\Lambda) +(\sigma-\sigma^{-1})\frac{\log(2\Lambda)\Gamma}{D}.
\]
It is readily checked that $\Lambda=\Lambda(\kappa)$ is necessarily diverging as $\kappa$ is small, which implies that
\[
\alpha  \log \frac1{\kappa} \le  \log \Lambda
\]
for some $\alpha\in(0,1)$, provided that $\kappa \le \kappa_0$  sufficiently small, where $\alpha$ and $\kappa_0$ depend on $\gamma_0, \Gamma, D$ and $\sigma$. This completes the proof.
\end{proof}

\section{Conclusions}
\label{sec:Conclusion}

To clarify how diffusion influences the ultimate efficiency of optimally constrained stirring, we exploit a shell model framework that makes it possible to probe the long-time behavior of optimal mixing in an advection-diffusion equation.  By dramatically reducing the complexity of the underlying dynamics while retaining the essential multiscale transfer mechanisms, the shell model allows systematic exploration of asymptotic regimes that are inaccessible to direct numerical simulations of the full PDEs.  Focusing on the decay of the scalar variance measured by the $H^{-1}$ norm and enforcing an enstrophy constraint on the stirring, we isolate the mechanisms that govern mixing efficiency in the asymptotic, long-time limit.

Our numerical results demonstrate that, when diffusion is present ($\kappa>0$), optimal stirring drives the scalar field toward a state in which the scalar length scale is arrested at a generalized Batchelor scale.  Once this regime is reached, the $H^{-1}$ mix-norm decays exponentially in time with a rate that is independent of the diffusivity.  This behavior is observed consistently under both local-in-time and global-in-time optimization strategies, indicating that the long-time mixing dynamics are robust with respect to the choice of optimization protocol; moreover, the global-in-time strategy achieves a strictly higher exponential decay rate than the local-in-time strategy for both $\kappa=0$ and $\kappa>0$.

A key finding of this study is that even arbitrarily small diffusion fundamentally alters the long-time asymptotic mixing behavior compared with the purely advective case ($\kappa=0$). While the no-diffusion system exhibits a slower decay of the $H^{-1}$ norm, the introduction of any positive diffusivity leads to a qualitatively different asymptotic regime characterized by a much faster, diffusivity-independent exponential decay. In this sense, diffusion acts as a singular perturbation in the long-time limit: although it does not control the decay rate itself, its presence is essential for enabling the enhanced asymptotic mixing observed here. Furthermore, increasing diffusivity improves mixing efficiency by reducing the prefactor of the exponential decay, thereby accelerating the approach to large-scale homogenization at finite times, even though the asymptotic decay rate remains unchanged.

Guided by these numerical observations, we derived new conditional lower bounds on the decay rate of the $H^{-1}$ norm that are strictly independent of the diffusivity for all $\kappa>0$. {We also established conditional upper bounds on the maximal rate of enhanced dissipation under the assumption that the $L^2$ norm decays at least exponentially in time.} 
These bounds rigorously capture the distinct asymptotic behavior of the diffusive system and confirm that the presence of diffusion, no matter how weak, enforces a fundamentally different long-time mixing regime. The decay rates measured in our computations are fully consistent with these theoretical bounds, providing strong validation of both the numerical and analytical approaches.

Beyond the shell model itself, the insights obtained in this study have broader implications. The clear separation between diffusive and non-diffusive asymptotic regimes, the identification of a Batchelor-scale-limited mixing state, and the diffusivity-independent decay bounds all point toward mechanisms that are expected to persist in more realistic PDE settings. 
As such, these results inform future studies of optimal mixing in PDE models, the development of sharp bounds on scalar transport, and the design of efficient stirring strategies in more complex fluid systems, including those governed by the Navier--Stokes equations. 
The shell model serves as a powerful and insightful proxy for the full PDE problem, offering a valuable stepping stone toward a deeper theoretical understanding of optimal mixing and scalar transport in fluid flows.

\section*{Acknowledgments}
The present work began when the first two authors were affiliated with the University of Michigan---the first as a doctoral student and the second as a postdoctoral researcher---under the supervision of Prof.~Charles R.\ Doering.  We dedicate this work to the memory of Charlie, who passed away in 2021 while this project was underway.  Charlie was an exceptional scientist and a deeply influential mentor.  His insight, rigor, and generosity shaped both the direction and the standards of this work, and his guidance continues to inspire us.  We are profoundly grateful for his mentorship and his enduring impact on
our work and on the field.

This work was supported by US National Science Foundation awards DMS-1515161, DMS-1813003, and DMS-2532634, the Deutsche Forschungsgemeinschaft (DFG, German Research Foundation) under Germany's Excellence Strategy EXC 2044/2--390685587, Mathematics M\"unster: Dynamics--Geometry--Structure, and computational resources and services provided by Advanced Research Computing at the University of Michigan.

%
%
%
%
%
\bibliographystyle{unsrt}
\bibliography{Mixing}


%
%
%

\end{document}